\documentclass[11pt]{article}

\usepackage{fullpage}

\usepackage[utf8]{inputenc}
\usepackage{amsmath}
\usepackage{amsthm}
\usepackage{thmtools}
\usepackage{thm-restate}
\usepackage{bbm}
\usepackage{amsfonts}
\usepackage{amssymb}
\usepackage{units}
\usepackage{xcolor}
\usepackage{graphicx}
\usepackage[ruled,linesnumbered]{algorithm2e}
\usepackage{comment}
\usepackage{xfrac}
\usepackage{tikz}
\let\oldnl\nl
\newcommand{\nonl}{\renewcommand{\nl}{\let\nl\oldnl}}
\newcommand{\etal}{\textit{et~al.}\xspace}

\declaretheorem[name=Theorem,numberwithin=section]{theorem}
\declaretheorem[sibling=theorem,name=Lemma]{lemma}
\declaretheorem[sibling=theorem,name=Claim]{claim}
\declaretheorem[sibling=theorem,name=Corollary]{corollary}

\declaretheorem[sibling=theorem,name=Definition]{definition}

\newcommand{\C}{\mathcal{C}}

\newcommand{\dist}{\mathrm{dist}}
\newcommand{\e}{\mathbb{E}}

\newcommand{\ld}{{\log{d}}}
\newcommand{\eps}{\varepsilon}

\newcommand{\spaOne}{\mbox{\sf SpannerPartOne(G,d,k) }}
\newcommand{\spaTwo}{\mbox{\sf SpannerPartTwo(G$_{\texttt{s}}$,R,k) }}

\providecommand{\email}[1]{\texttt{#1}}
\providecommand{\poly}{\operatorname{poly}}

\title{$(\alpha, \beta)$ Spanners and Hybrid Spanners with Nearly Tight Bounds
\thanks{A preliminary version of this paper appeared in SODA 2026.}}

\author{
Shiri Chechik
\and
Gur Lifshitz%
\thanks{Tel Aviv University, Israel
(\email{shiri.chechik@gmail.com},
\email{gurlifshitz@gmail.com}).}
}

\begin{document}

\maketitle

\begin{abstract}
An $(\alpha,\beta)$-spanner of an $n$-vertex undirected, unweighted graph $G=(V,E)$ is a subgraph $H$ satisfying, for all $u,v\in V$,
\[
\dist_G(u,v) \le \dist_H(u,v) \le \alpha\cdot \dist_G(u,v) + \beta.
\]
For any $k\in\mathbb{N}$, classical results show that a $(2k-1,0)$-spanner with $O(n^{1+1/k})$ edges exists and is asymptotically optimal under Erd\H{o}s' girth conjecture. This conditional lower bound comes from adjacent pairs and does not rule out better stretch for more distant pairs.

We present new spanner constructions that achieve nearly optimal guarantees for all distances $d \le k$. Specifically, we construct a spanner $H\subseteq G$ with
\(
O( n^{1+1/k} + (k+d \log{d}) \, n)
\)
edges, ensuring that any pair at original distance at most $d$ satisfies $\dist_H(u,v)\le 2k + O(d\log d)$. Equivalently, the multiplicative stretch for pairs at distance $d$ is $2k/d + O(\log d)$. 

In particular, setting $d=k/\log k$ yields an $(O(\log k),O(k))$-spanner with $O(n^{1+1/k}+kn)$ edges. For comparison, Ben-Levy and Parter (SODA'20) obtained, for every fixed $\varepsilon>0$ and sufficiently large $k$, an $(O(k^\varepsilon),O_\varepsilon(k))$-spanner with $O_{\varepsilon,k}(n^{1+1/k})$ edges. Our result improves the multiplicative stretch from $O(k^\varepsilon)$ to $O(\log k)$ while keeping the additive term linear in $k$, bringing us closer to the goal of $(O(1),O(k))$-spanners.

Furthermore, Ben-Levy and Parter obtained multiplicative stretch $O_\varepsilon(k/d)$ for distances $d\le k^{1-\varepsilon}$, for every fixed $\varepsilon>0$, and an explicit bound of $7k/d$ for $d\le\sqrt{k}/2$. We achieve $2k/d+O(\log d)$, which is $(2+o(1))k/d$ whenever $d=o(k/\log k)$.

Our second result is an improved construction of $k$-hybrid spanners, which guarantee stretch $2k-1$ for adjacent pairs and $k$ for non-adjacent pairs. Parter's original construction uses $O(k^2 n^{1+1/k})$ edges; we achieve the same guarantees with $O(n^{1+1/k}+kn)$ edges, removing the $k^2$ factor from the $n^{1+1/k}$ term. For every fixed $k$, our edge bound is optimal up to a constant factor under Erd\H{o}s' girth conjecture.
\end{abstract}

\section{Introduction}

Graph spanners are fundamental combinatorial structures that yield compact graph representations while approximately preserving pairwise distances. Since their introduction in \cite{peleg1988graph, Optimal_Synchronizer}, spanners have found widespread applications in network design, distributed computing, routing, and graph sparsification \cite{Distributed_Computing}. Given an undirected, unweighted $n$-vertex graph $G=(V,E)$ and a function $f:\mathbb{N}\to\mathbb{N}$, a subgraph $H\subseteq G$ is called an \emph{$f(d)$-spanner} if, for every pair $(u,v)\in V\times V$, it holds that
\[
\dist_H(u,v)\ \le\ f\!\bigl(\dist_G(u,v)\bigr).
\]
We refer to $f$ as the \emph{stretch function}; there is a tradeoff between the size of $H$ and the quality of $f$.

The study of spanners began with \emph{$\alpha$-multiplicative} spanners, i.e., when $f(d)=\alpha\cdot d$ for some $\alpha > 1$. For every $k\in\mathbb{N}$ it is well known that a $(2k\!-\!1)$-multiplicative spanner with $O(n^{1+1/k})$ edges can be efficiently constructed; this bound is optimal under Erd\H{o}s' girth conjecture \cite{P1964, sparse_spanners}. 

Another line of work studies \emph{additive} spanners, i.e., $f(d)=d+\beta$. There are classical constructions achieving constant additive stretch $\beta\in\{2,4,6\}$ with progressively sparser graphs as $\beta$ increases; see \cite{plusTwo, Chechik2016, alpha_beta_05} for the exact bounds and tradeoffs. When one insists on very sparse spanners, a \emph{phase transition} occurs: Abboud and Bodwin~\cite{4/3tight} showed that, for every fixed $\delta>0$, spanners with $O(n^{4/3-\delta})$ edges require additive error $n^{\Omega_\delta(1)}$ in the worst case. Tan and Zhang~\cite{Tan23} constructed additive spanners with $\widetilde{O}(n)$ edges and additive error $O(n^{0.403})$, that improves the error exponent of the linear-size spanners of Bodwin and Vassilevska Williams~\cite{Bod_wil21}, whose additive error is $O(n^{3/7+\varepsilon})$ for any fixed $\varepsilon>0$.

These directions motivated the study of \emph{$(\alpha,\beta)$-spanners}, which combine both multiplicative and additive terms, i.e., $f(d)=\alpha\cdot d+\beta$. Elkin and Peleg \cite{ElkinPeleg01} showed that, for every $\varepsilon\in(0,1)$ and integer $k\ge 1$, there exist $(1+\varepsilon,\beta)$-spanners with $O_{\varepsilon,k}(n^{1+1/k})$\footnote{The notation \(O_a(\cdot)\) hides a multiplicative dependence on a function of \(a\).} edges and
\(
\beta \;=\; \left(O\!\left(\tfrac{\log k}{\varepsilon}\right)\right)^{\log k},
\)
thereby establishing explicit improved guarantees for sufficiently distant pairs. Thorup and Zwick~\cite{TZ06} gave a simple sparse \emph{emulator}\footnote{An $f(d)$-emulator is a (possibly non-subgraph) weighted sparse graph $G'$ that preserves all pairwise distances up to a stretch function $f(d)$, i.e., $\dist_G(u,v) \le \dist_{G'}(u,v) \le f(\dist_G(u,v))$ for all $u,v \in V$. It is thus a relaxed variant of a spanner.} that is a $(1+\varepsilon,\beta)$-emulator for every $\varepsilon>0$, where smaller values of $\varepsilon$ entail larger additive terms $\beta$. Expressed in terms of our sparsity parameter $k$, their construction has $O_k(n^{1+1/k})$ edges and sublinear additive stretch
\(
f(d)=d+O\!\bigl(\log k\cdot d^{\,1-1/\log k}+\poly(k)\bigr).
\)
The same emulator satisfies this bound simultaneously for all distances $d$.
Abboud, Bodwin, and Pettie~\cite{Lower_Bounds} showed that spanners with $O_k(n^{1+1/k})$ edges must incur additive error $\Omega_k\!\bigl(d^{\,1-O(1/\log k)}\bigr)$ for some pairs at distance $d$. Moreover, for sufficiently small $\varepsilon>0$, any $(1+\varepsilon,\beta)$-spanner construction with $O_{\varepsilon,k}(n^{1+1/k})$ edges must satisfy
\(
\beta=\left(\frac{1}{\varepsilon\log k}\right)^{\Omega(\log k)}.
\)
These lower bounds do not exclude $(1+\varepsilon,O_\varepsilon(k))$-spanners with $O_{\varepsilon,k}(n^{1+1/k})$ edges.

\medskip
\noindent\textbf{Other $(\alpha,\beta)$ regimes.}
While the near-additive spanners of Elkin and Peleg~\cite{ElkinPeleg01} provide a $(1+\varepsilon)$ multiplicative approximation for sufficiently large distances, a natural question is how well one can approximate \emph{smaller} distances. Ben-Levy and Parter~\cite{New_ab_spanners} addressed this regime by constructing spanners that achieve improved stretch guarantees for shorter distances while maintaining near-optimal sparsity of $O_{\varepsilon,k}(n^{1+1/k})$ edges. They obtained (up to constants) nearly optimal bounds for more distance ranges, achieving $f(d) = O_\varepsilon(k + d)$ outside the window $d \in [k^{1-\varepsilon},\,k^{1+\varepsilon}]$, with concrete instantiations such as $f(d) = 7k$ for $d \le \sqrt{k}/2$ and an $(O(k^\varepsilon), O_\varepsilon(k))$-spanner, implying $f(k^{1-\varepsilon}) = O_\varepsilon(k)$. Moreover, they constructed $(\alpha,\beta)$-spanners with $\alpha = O_\varepsilon(1)$ and $\beta = k^{1+\varepsilon}$, providing constant multiplicative stretch from distance $d \ge k^{1+\varepsilon}$. These results delineate the current picture under near-optimal sparsity and highlight the remaining open gap around $d \approx k$.


The ``holy grail'' in this line of research is to achieve a stretch function $f(d) = 2k + O(d)$ while maintaining the optimal $O(n^{1+1/k})$ edge bound. When $d \approx k$, this corresponds to constructing an $(O(1), O(k))$-spanner. This target is believed to be the best possible: under the girth conjecture, there exists a graph $G$ and a pair of vertices at distance $d$ in $G$ that must be at least $2k + O(d)$ apart in any spanner $H \subsetneq G$. Moreover, by a classical lower bound of Woodruff~\cite{woodruff2006}, who showed that for every fixed $k$, guaranteeing $f(k)<3k$ requires $\Omega(k^{-1}n^{1+1/k})$ edges in the worst case.

\medskip
A notable related result is due to Baswana, Kavitha, Mehlhorn, and Pettie \cite{alpha_beta_05}, who constructed a $(k,k\!-\!1)$-spanner that guarantees multiplicative stretch $2k\!-\!1$ for neighboring vertices and at most $3k/2$ for non-adjacent pairs. Another approach to improving stretch guarantees is
through hybrid spanners. Parter \cite{Bypassing} introduced $k$-hybrid spanners, which achieve multiplicative stretch $2k\!-\!1$ for adjacent pairs while reducing the stretch to $k$ for non-adjacent pairs, using $O(k^2 n^{1+1/k})$ edges. While optimal for pairs at distance $2$ under the girth conjecture, this construction provides no substantial improvement for distant pairs.

A substantial body of work explores additional tradeoffs among stretch, sparsity, and computational efficiency. See, for example, \cite{tutorial, NoC235, Baswana_linear_weighted, Low_distortion_spanners, Very_Sparse_Spanners, Additive_Spanners}. For directed-graph spanners, see \cite{Roundtrip_spanners08, roundtrip_spanners17, roundtrip_spanners18, roundtrip_spanners19, roundtrip_spanners20a, roundtrip_spanners20b, roundtrip_spanners21, roundtrip_spanners23}.

\subsection{Our Results}

We show that for any $d\le k$, there exists a spanner with $O_k(n^{1+1/k})$ edges such that any pair of vertices at distance $\leq d$ in $G$ is at most $2k + O(d \log d)$ apart in $H$. This brings us significantly closer to the conjectured lower bound of \(2k + O(d)\) under the Erd\H{o}s' girth conjecture.

\begin{theorem}
\label{th:ab-spa}
Let $G = (V, E)$ be an undirected unweighted graph, and let $k,d$ be integers such that $ d \leq k$. Then, there is a randomized algorithm that runs in $O(m\sqrt{n})$ expected time, constructing a spanner $H$ with $O\!\bigl(n^{1+1/k} + (k + d\log d)\,n\bigr)$ edges in expectation, such that for every pair $(u,v)$ with $\dist_G(v,u)\le d$,
\[
\dist_H(v,u) \le 2k + O(d\log d).
\]
\end{theorem}

This result can be compared to that of Ben-Levy and Parter~\cite{New_ab_spanners}, who showed that for any \(0 < \varepsilon < \tfrac{1}{2}\) and integers \(k > 2^{4/\varepsilon}, d \ge 1\), one can construct a spanner \(H \subseteq G\) of expected size
\(
O\!\left(k^\varepsilon n^{1+1/k} + \bigl(2^{6/\varepsilon}k^{1+\varepsilon} + d\bigr)n\right),
\)
such that every pair of vertices at distance at most \(d\) in \(G\) satisfies
\(
\dist_H(u,v) \le O\!\left(2^{6/\varepsilon}k + k^\varepsilon d\right).
\)
Their global construction gives an $(O(k^\varepsilon),2^{O(1/\varepsilon)}k)$-spanner with $O_{\varepsilon,k}(n^{1+1/k})$ edges.

Our construction achieves strictly stronger guarantees: for all \(d \le k/\log k\), we obtain an \((O(k/d), O(k))\)-spanner using fewer edges and with faster expected running time. Substituting \(d = k/\log k\) in Theorem~\ref{th:ab-spa} yields an \((O(\log k), O(k))\)-spanner, which approaches the conjectured \((O(1), O(k))\) lower bound when \(d \approx k\). Thus, compared with Ben-Levy and Parter’s bounds in the intermediate range \(d \in [k^{1 - o(1)},\,k^{1 + o(1)}]\), our results achieve an exponentially smaller multiplicative stretch (while maintaining the same additive term), thereby filling a significant gap in the literature.

\begin{corollary}
\label{co:logk-k-spa}
For any integer $ k $ and an undirected unweighted graph $G = (V, E)$, there exists a $(O(\log k), O(k))$-spanner $H$ with $ O(n^{1 + 1/k} + kn)$ edges.
\end{corollary}

By combining all our constructions for distances $d\le k$ with the bounds of Ben-Levy and Parter~\cite{New_ab_spanners}, we obtain the following piecewise upper bound on the stretch function.

\begin{corollary}
For any $\varepsilon>0$, integer $k$, and any undirected unweighted graph $G=(V,E)$, there exists an $f(d)$-spanner $H$ with $O_{\varepsilon,k}\!\left(n^{1+1/k}\right)$ edges, where
\[
f(d)\;=\;
\begin{cases}
2k + O\!\bigl(d\log d\bigr), & \text{if } d \le k,\\[2pt]
O\!\bigl(d\log k\bigr), & \text{if } k < d \le k^{1+\varepsilon},\\[2pt]
O_\varepsilon\!\bigl(d\bigr), & \text{if } d > k^{1+\varepsilon}.
\end{cases}
\]
\end{corollary}

In addition to our main spanner construction, we present a $k$-hybrid spanner with $O(n^{1 + 1/k}+kn)$ edges, improving upon Parter's construction \cite{Bypassing}, which required $O(k^2 n^{1 + 1/k})$ edges. For every fixed $k$, this matches the optimal exponent of $n$ under the girth conjecture, while removing the $k^2$ multiplicative factor from Parter's edge bound.

\begin{theorem}
\label{th:2-spa}
Let $G = (V, E)$ be an undirected unweighted graph, and let $k \in \mathbb{N}$. Then, there is a randomized construction of a spanner $H$ with $O(n^{1 + 1/k}+kn)$ edges in expectation that provides a multiplicative stretch of $2k - 1$ for every pair of neighboring vertices $u$ and $v$ and a multiplicative stretch of $k$ for the rest of the pairs. This construction can be done in $ O(m \sqrt{n})$ expected time.
\end{theorem}

\renewcommand{\arraystretch}{1.3}
\begin{table}[h!]
\centering
\small
\begin{tabular}{|p{4cm}|p{5.5cm}|p{5cm}|}
\hline
$\mathbf{f(d) = }$
  & \textbf{Size (edges)}
  & \textbf{Notes / Parameters} \\
\hline
\hline

$(2k-1)\cdot d$
  & $O(n^{1+1/k})$
  & Classic $(2k-1,0)$-spanner \\
\hline

$d + O_k\!\left(d^{1-\Theta(1/\log k)}\right)$
  & $O_k(n^{1+1/k})$
  & fixed $k\ge 2$ \cite{Tan23} \\
    \hline

$O(k^\varepsilon\cdot d + 2^{6/\varepsilon}\cdot k)$
  & $O\bigl(k^\varepsilon n^{1+1/k}
      +(2^{6/\varepsilon}k^{1+\varepsilon}+d)n\bigr)$
  & $0<\varepsilon<\frac12$, $k>2^{4/\varepsilon}$ \cite{New_ab_spanners} \\
\hline

$(3+\varepsilon)d
  +O\!\left(\tfrac{1}{\varepsilon}
    k^{\log(5+16/\varepsilon)}\right)$
  & $O\!\left(n^{1+1/k}
      +\tfrac{1}{\varepsilon}
       k^{\log(5+16/\varepsilon)}n\right)$
  & $0<\varepsilon<1$;
    $(3+\varepsilon,\beta)$-spanner
    \cite{New_ab_spanners} \\
\hline

$2k+O(d\log d)$
  & $O\bigl(n^{1+1/k}+(k+d\log d)n\bigr)$
  & \textbf{Theorem \ref{th:ab-spa}} \\
\hline

$2k+O(d)$
  & $O\bigl(n^{1+1/k}\bigr)$
  & $(O(1),O(k))$-spanner, \textbf{Conjectured} \\
\hline

When $d\le \frac{k}{\log k}$, $O(k)$
  & $O\bigl(n^{1+1/k}+kn\bigr)$
  & $(O(\log k),O(k))$-spanner,
    \textbf{Corollary \ref{co:logk-k-spa}} \\
\hline
\hline

When $d\ge 2$, $d\cdot k$
  & $O(k^2n^{1+1/k})$
  & $k$-hybrid spanner \cite{Bypassing} \\
\hline

When $d\ge 2$, $d\cdot k$
  & $O(n^{1+1/k}+kn)$
  & $k$-hybrid spanner,
    \textbf{Theorem \ref{th:2-spa}} \\
\hline
\end{tabular}
\caption{Distance bound $f(d)$ for pairs with $\dist_G(u,v)=d$,
edge count and key assumptions/parameters.}
\label{tab:spanner-summary}
\end{table}

\paragraph{Paper Organization.}
Section~\ref{sec:prelim} provides preliminary definitions and notations. Section~\ref{sec:ab-spa} presents the main spanner construction and proves its size and stretch guarantees, as stated in Theorem~\ref{th:ab-spa}. Section~\ref{sec:h-spa} describes the hybrid spanner construction and proves Theorem~\ref{th:2-spa}. Finally, Appendix~\ref{sec:apndx} shows how to implement the construction of Theorem~\ref{th:ab-spa} in expected time \(O(m \sqrt{n})\).

\section{Preliminaries} \label{sec:prelim}

We consider an undirected, unweighted graph $G = (V, E)$. Given a subgraph $H = (V(H), E(H))$ and a subset of vertices $A \subseteq V(H)$, we define the edge set $E_H(x, A) \subseteq E(H)$ as the set of edges with one endpoint at $x$ and the other in $A$, formally:
\(
E_H(x, A) = \{ (x, u) \in E(H) \mid u \in A \}.
\)
Furthermore, for any integer \( r \), we define \( B_H(x,r) \) as the set of vertices within distance at most \( r \) from \( x \) in \( H \), i.e.,  
\(
B_H(x,r) = \{ u \in V \mid \dist_H(x,u) \leq r \}.
\)
For brevity, we occasionally omit the graph subscript when the context is clear.

For a subset $V_i \subseteq V$, we denote by $G[V_i]$ the subgraph of $G$ induced by $V_i$. In particular, when referring to the induced subgraph $G[V_i]$, we use $\dist_i(v, u)$ as shorthand for $\dist_{G[V_i]}(v, u)$.

In the following sections, we occasionally perform a BFS from a set of vertices $C \subseteq V_i$. This can be implemented by adding a dummy vertex $s$ connected to every vertex in $C$, and then running a standard BFS from $s$ in the subgraph $G[V_i]$. This procedure runs in $O(m+n)$ time and produces a BFS forest in which each vertex of $G[V_i]$ is connected to its closest vertex in $C$ (ties between equidistant vertices may be broken arbitrarily).

\vspace{1em}
In this paper, we consider the following types of spanners:
\begin{definition}[$(\alpha, \beta)$-spanners]
    Given an undirected graph $G = (V, E)$, a subgraph $H \subseteq G$ is called an $(\alpha, \beta)$-spanner if, for every pair of vertices $(u, v) \in V \times V$, it holds that:
\begin{equation*}
    \dist_H(u, v) \leq \alpha \cdot \dist_G(u, v) + \beta.
\end{equation*}
\end{definition} 

\begin{definition}[$k$-hybrid spanners]
For a given graph $G = (V, E)$, a subgraph $H \subseteq G$ is called a $k$-hybrid spanner if, for every pair of vertices $(u, v) \in V \times V$, the following holds:
\begin{equation*}
    \dist_H(u, v) \leq 
    \begin{cases}
        (2k - 1) \cdot \dist_G(u, v), & \text{if } (u, v) \in E(G); \\
        k \cdot \dist_G(u, v), & \text{otherwise}.
    \end{cases}
\end{equation*}
\end{definition}

\section{The $ 2k + O(d\ld) $ Spanner Construction} \label{sec:ab-spa}

This section establishes Theorem~\ref{th:ab-spa} by presenting a randomized construction and analyzing its stretch and size guarantees. Given an undirected, unweighted graph \(G = (V, E)\) and two integers \(k, d \in \mathbb{N}\) with \(d \leq k\), the construction outputs a spanner \(H\) that, in expectation, has \(O(n^{1+1/k} + (k + d \log d) n)\) edges, and ensures that for every pair of vertices at distance \(\leq d\) in \(G\), their distance in \(H\) is at most \(2k + O(d \log d)\). Moreover, by setting \(d = \tfrac{k}{\log k}\), the spanner \(H\) achieves stretch parameters \((O(\log k), O(k))\) with \(O(n^{1+1/k} + kn)\) edges.

We defer the running time analysis to Appendix~\ref{sec:apndx}, where we describe a modified version of the algorithm that achieves an expected runtime of \(O(m \sqrt{n})\).

We begin by assuming that $d$ is a power of $2$ and that $d$ divides $k$. If that is not the case, let $\hat{d} \in [d, 2d)$ be the smallest power of $2$ that is larger than $d$, and $\hat{k} \in [k, k + \hat{d})$ the smallest integer such that $\hat{d}$ divides $\hat{k}$. Running the algorithm with $\hat{d}, \hat{k}$, ensures that the resulting spanner $H$ has size
\(O(n^{1+1/\hat{k}}+(\hat{k}+\hat{d}\log\hat{d})n)=O(n^{1+1/k}+(k+d\ld)n)\).
Moreover, for every pair at distance at most $d$ in $G$, their distance in $H$ is at most
\(2\hat{k}+O(\hat{d}\log\hat{d})=2k+O(d\ld)\). We will revisit this adjustment at the end of the section and provide the explicit constants.

\subsection{High-Level Intuition (Setting \(d = k\))}

We start with a warm-up that shows a construction of a spanner $H$ that provides an \(O(\log k)\)-multiplicative stretch for every path of length \(k\), using only \(O(n^{1+1/k} + k \log k \cdot n)\) edges. In other words, $H$ is an $(O(\log k), O(k \log k))$-spanner. This warm-up already improves previous results which, for any constant \(\eps>0\) and \(k>2^{4/\eps}\), construct a spanner \(H\) with \(O(k^\eps)\)-multiplicative stretch for every path of length \(k\), using $O_{\epsilon,k}(n^{1+1/k})$ edges \cite{New_ab_spanners}.

In this warm-up, we focus on the intuition behind the construction and omit formal proofs. The general construction for arbitrary $d$ builds upon these ideas. More precisely, it consists of two phases, where the first phase is very close to the algorithm presented in this warm-up.

We initialize \(H = (V, \emptyset)\) and set \(V_0 = V\).  For a general \(i\), we write
\[
  \dist_i(u,v) \;:=\; \dist_{G[V_i]}(u,v),
  \quad
  B_i(x,r) \;:=\; \{\,y\in V_i : \dist_i(x,y)\le r\},
\]
so that all distances and balls in iteration \(i\) are taken with respect to the induced subgraph \(G[V_i]\), where $V_i$ (formally defined later) is the set of vertices that are still uncovered in iteration $i$. That is, vertices $v$ for which there may still exist paths of length $k$ that contain $v$, and the constructed spanner does not yet contain a suitable alternative path. The algorithm proceeds for exactly \(\log k\) iterations, assuming that \(k\) is a power of two. In each iteration, the algorithm randomly selects a set of \emph{centers}, and adds edges to $H$ in two ways: (1) it adds a shortest path between every two centers that are sufficiently close and (2) it adds a shortest path between every vertex to its closest center. Afterwards, it removes from the graph all vertices whose distance to their nearest center is below some threshold.
We begin by describing the process in the first iteration, and then generalize to subsequent iterations.

\paragraph{Iteration 0.}
The algorithm samples a random set \(C_0 \subseteq V_0\), where each vertex is selected independently with probability \(1/\sqrt{n}\). Thus, \(\mathbb{E}[|C_0|] = \sqrt{n}\). It updates the spanner \(H\) in two steps:

\begin{enumerate}
    \item[$\bullet$] \textbf{Link centers:} For each pair \(c, c' \in C_0\), if their distance in \(G[V_0] (=G)\) is at most \(r_0 = O(k \log k)\), it adds the shortest path between them to \(H\).
    
    \item[$\bullet$] \textbf{Link vertices to centers:} It performs a BFS from the vertex set \(C_0\) over \(G[V_0]\), and adds the resulting BFS forest to \(H\). This means that for every vertex \(v \in V_0\), the spanner \(H\) contains a shortest path from \(v\) to its closest center in \(C_0\) (ties are broken arbitrarily). 
\end{enumerate}

These steps add at most \(O(k\log k\cdot n)\) edges to \(H\): there are at most \(|C_0|^2=O(n)\) eligible pairs, each contributing a path of length \(O(k\log k)\). The BFS adds \(O(n)\) more edges.

\paragraph{Stretch Analysis.}
Define \( r_1 = \tfrac{1}{2}r_0 - k \); intuitively, this is slightly less than half the radius threshold \( r_0 \) used to link centers.

Let \( P \subseteq G[V_0] \) be a path of length \( |P| = k \), connecting two vertices \( s, t \in V_0 \). Suppose there exists a vertex \( x \in P \) such that \( \dist_0(x, C_0) \leq r_1 \). We claim that \(P\) already has an \( O(\log k) \)-stretch in the spanner \(H\).

For simplicity, we illustrate the argument assuming \( x = s \). The general case, where \(x\) lies anywhere along the path \(P\), is similar but slightly more technical; we defer its full justification to the formal analysis later on. 

Assume \( \dist_0(s, C_0) \leq r_1 \). Let \( c_0(s), c_0(t) \in C_0 \) be the closest centers to \(s\) and \(t\), respectively, as established by the \textbf{Link vertices to centers} step. Then:
\begin{enumerate}
    \item \(\dist(t, c_0(t)) \leq \dist(t, c_0(s)) \leq \dist(t,s) + \dist(s, c_0(s)) \leq k + r_1 = \tfrac{1}{2}r_0\),
    
    \item So
    \[\dist(c_0(s), c_0(t)) \leq  \dist(c_0(s), s) + \dist(s,t)+ \dist(t, c_0(t)) \leq r_1 + k + \tfrac{1}{2}r_0 = r_0,\]
    which means that the shortest path between \(c_0(s)\) and \(c_0(t)\) was added to the spanner \(H\) during the \textbf{Link centers} step. This was exactly the reason for setting $r_1 = \tfrac{1}{2}r_0 - k$.
    \item Therefore, we get:
    \[
    \dist_H(s, t) \leq \dist(s, c_0(s)) + \dist(c_0(s), c_0(t)) + \dist(c_0(t), t) \leq 4r_1 + 3k = O(k \log k) = O(\log k)\cdot |P|.
    \]
\end{enumerate}

See figure \ref{fig:intuition}.

\begin{figure}[h!]
    \centering
\begin{tikzpicture}[scale=1.5, every node/.style={font=\small}]

\coordinate (s) at (0,0);
\coordinate (t) at (6,0);
\coordinate (cs) at (-1,2);
\coordinate (ct) at (7,2);

\draw[thick, black] (s) -- (cs) node[midway, left] {$\leq r_1$};
\draw[thick, black] (cs) -- (ct) node[midway, above] {$\leq r_1 + k + (r_1+k) = r_0$};
\draw[thick, black] (ct) -- (t) node[midway, right] {$\leq r_1+k$};
\draw[thick, magenta, dashed] (s) -- (t) node[midway, below right, magenta] {$k$};

\filldraw[black] (s) circle (2pt) node[below] {$s$};
\filldraw[black] (t) circle (2pt) node[below] {$t$};
\filldraw[black] (cs) circle (2pt) node[above] {$c_0(s)$};
\filldraw[black] (ct) circle (2pt) node[above] {$c_0(t)$};

\end{tikzpicture}
\caption{Let $P$ be a path between $s,t$ of length $k$. Assuming that $\dist(s,c_0(s)) \leq r_1$, we can prove that the path $P' = s \to c_0(s) \to c_0(t) \to t$ is entirely in $H$ and of length $ |P'| \leq 4r_1 + 3k$. The black paths are in $H$, and the dashed is in $G$.}
\label{fig:intuition}

\end{figure}
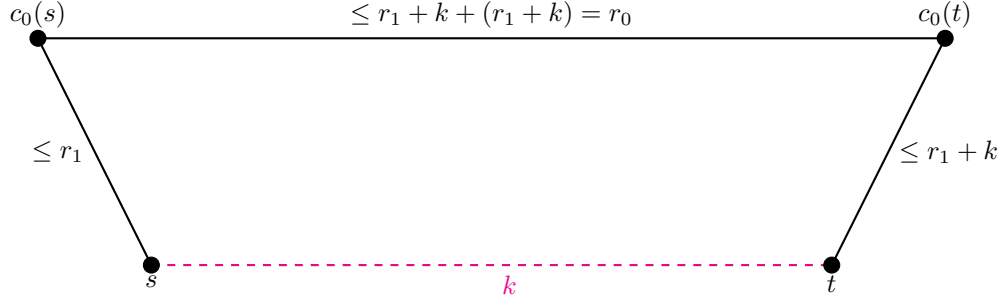

We conclude that every $x \in V_0$ such that $\dist_0(x, C_0) \leq r_1$ can already be considered \emph{handled} — all paths of length \(k\) passing through it are well-approximated in the spanner. Thus, the final step of iteration \(0\) is:

\paragraph{Prune Step.}
\begin{enumerate}
    \item[$\bullet$] \textbf{Prune:} The algorithm removes all such vertices and defines \(V_1 \subseteq V_0\) to be the set of remaining vertices, i.e:
    \(
    V_1 = \{ x\in V_0 \mid \dist_0(x, C_0) > r_1 \}.
    \)
\end{enumerate}

The algorithm sampled \(O(\sqrt{n})\) centers during this iteration and removed all vertices \(x\) for which \(\dist(x, C_0) \leq r_1\). This implies that the remaining vertices in \(V_1\) are sparse with respect to radius \(r_1\); that is, for every \(x \in V_1\), the ball \(B_1(x, r_1)\) contains at most \(O(\sqrt{n})\) vertices in expectation.

This sparsity allows the algorithm to gradually sample more and more centers in later iterations while still keeping the number of added edges small. For example, in iteration 1, it samples \(|C_1| = O(n^{3/4})\) centers, and connects each pair \(c, c' \in C_1\) if \(\dist_1(c, c') \leq r_1\). Although this is a larger number of centers than in the previous iteration, the remaining vertices in \(V_1\) are sparse: for each center \(c \in C_1\), the ball \(B_1(c, r_1)\) contains only \(O(\sqrt{n})\) vertices in expectation. Hence the number of eligible center pairs is \(O(n)\) in expectation, and each contributes a path of length \(O(r_1)\), yielding \(O(r_1 n)\) an addition of edges in expectation to $H$.

Generally, in iteration \(i < \log k\), it samples each vertex in \(V_i\) independently with probability \(n^{-1/2^{i+1}}\), yielding a center set \(C_i\) of expected size \(n^{1 - 1/2^{i+1}}\). Despite the growing number of centers, the expected number of edges added to \(H\) remains \(O_k(n)\).

\paragraph{Iteration \(i\).}
This process continues for \(i = 0,1,\dots,\log k - 1\), where \(V_i\) denotes the current vertex set. The purpose of iteration \(i\) is to approximate paths of length at most \(\tfrac{k}{2^i}\) in the subgraph \(G[V_i]\).

\begin{itemize}
    \item \textbf{Sample $C_i$:} It sets $C_i$ to be the set obtained by sampling each vertex in \(V_i\) independently with probability \(n^{-1/2^{i+1}}\), so \(|C_i| \leq n^{1 - 1/2^{i+1}}\) in expectation.
    \item \textbf{Link centers:} For every two vertices $c,c' \in C_i$ whose distance is at most $r_i$ in $G[V_i]$, it adds a shortest path between them to $H$.
    \item \textbf{Link vertices to centers:} It performs a BFS from the vertex set \(C_i\) over \(G[V_i]\), and adds the BFS forest to $H$.
    \item \textbf{Prune:} It defines radius \(r_{i+1} = \tfrac{1}{2}r_i - \tfrac{k}{2^i}\), and removes all vertices \(x \in V_i\) with \(\dist_i(x, C_i) \leq r_{i+1}\), and sets \(V_{i+1}\) to be the remaining vertex set (where $\dist_i := \dist_{G[V_i]}$).
\end{itemize}

We choose the recursive relation \(r_{i+1} = \tfrac{1}{2}r_i - \tfrac{k}{2^i}\) to ensure that, in each iteration, any path \(P \subseteq G[V_i]\) of length at most \(\tfrac{k}{2^i}\) that passes through a vertex \(x \in P\) with \(\dist_i(x, C_i) \leq r_{i+1}\), can be approximated in \(H\) via a detour through the centers with only an \(O(\log k)\)-multiplicative stretch. 

Specifically, let \(s,t \in P\) be the endpoints of the path, and let \(c_i(s), c_i(t) \in C_i\) be the closest centers to \(s\) and \(t\), respectively. Since \(x\) lies within distance \(r_{i+1}\) of a center and \(|P| \leq \tfrac{k}{2^i}\), one can prove that 
\( \dist_i(c_i(s), c_i(t)) \leq r_i. \)
Thus, during the \textbf{Link centers} step, a shortest path between \(c_i(s)\) and \(c_i(t)\) was added to \(H\). The detour path \(P' := s \to c_i(s) \to c_i(t) \to t\) is entirely contained in \(H\) and provides an \(O(\log k)\)-stretch for \(P\).

Here, we focus on approximating paths of length \(\tfrac{k}{2^i}\) rather than \(k\), because the algorithm repeats this process for \(\log k\) iterations. If we were to approximate paths of length \(k\) in each iteration, the recursive relation for the radius parameter would become \(r_{i+1} = \tfrac{1}{2}r_i - k\). Even assuming \(r_{\log k} = 0\), this would imply that \(r_0 = O(k^2)\), which is prohibitively large. 
To avoid this blow-up, we progressively targeted shorter paths in each iteration.

\paragraph{Final Step.}
After \(\log k\) iterations, we reach a set \(V_{\log k}\) in which, for every vertex \(x\), the ball \(B_{\log k}(x, r_{\log k})\) contains at most \(\frac{n}{|C_{\log k-1}|}=O(n^{1/k})\) vertices in expectation. Setting \(r_{\log k} = 1\), we ensure that the degree of each vertex in \(G[V_{\log k}]\) is at most \(O(n^{1/k})\) in expectation, so the algorithm simply adds all edges in \(G[V_{\log k}]\) to \(H\), completing the construction.

Altogether, over \(\log k\) rounds, \(\sum_i O(r_i n)=O(k\log k \cdot n)\) edges are added during \textbf{Link centers} steps plus $O(\log k \cdot n)$ during \textbf{Link vertices to centers} steps, plus $O(n^{1+1/k})$ in the last cleanup, giving overall $O(n^{1+1/k} + k \log k \cdot n)$ edges that are in $H$.

\paragraph{Stretch Analysis.} In each iteration \(i < \log k\), we constructed \(H\) to guarantee an \(O(\log k)\)-multiplicative stretch for every path of length \(\tfrac{k}{2^i}\) that includes a vertex \(x \in V_i \setminus V_{i+1}\), i.e., one that was removed during that iteration.

We now argue that \emph{every} path of length \(\tfrac{k}{2^i}\) in the subgraph \(G[V_i]\) is approximated in \(H\) with stretch \(O(\log k)\). We prove this by induction, going from \(i = \log k\) down to \(i = 0\). Notably, the case \(i = 0\) implies that every path of length \(k\) in \(G\) is approximated within an \(O(\log k)\) factor in \(H\), as desired.

\smallskip
\noindent\textbf{Base case (\(i = \log k\)):} At the final iteration, we added all edges of the induced subgraph \(G[V_{\log k}]\) to \(H\), so every path in \(G[V_{\log k}]\) is preserved exactly.

\smallskip
\noindent\textbf{Inductive step:} Assume that for iteration \(i+1\), every path of length \(\tfrac{k}{2^{i+1}}\) in \(G[V_{i+1}]\) has an \(O(\log k)\)-stretch in \(H\). Consider an arbitrary path \(P \subseteq G[V_i]\) of length \(\tfrac{k}{2^i}\). We distinguish two cases:
\begin{itemize}
  \item If \(P \subseteq G[V_{i+1}]\), then we can partition \(P\) into two subpaths, each of length \(\tfrac{k}{2^{i+1}}\). By the inductive step, each subpath is approximated in \(H\) with stretch \(O(\log k)\), so their concatenation gives an \(O(\log k)\)-approximate path for \(P\).
  
  \item Otherwise, \(P\) contains some vertex \(x \in V_i \setminus V_{i+1}\) that was pruned in iteration \(i\), i.e., \(\dist_i(x, C_i) \leq r_{i+1}\). But then, by the covering guarantees of that iteration, the entire path \(P\) is approximated in \(H\) with \(O(\log k)\) stretch.
\end{itemize}

This completes the inductive proof that all length-\(k\) paths are well-approximated in \(H\), establishing the claimed stretch bound. One can show that $H$ is an $(O(\log k), O(k \cdot \log k))$ spanner.

\medskip

\paragraph{The Construction Overview}

Our algorithm consists of two phases. The first phase closely follows the high-level intuition described earlier, with a few necessary adjustments. In the general setting, our goal is to approximate all paths of length at most $d \leq k$. Accordingly, the first phase now proceeds for $\log d$ iterations, where in iteration $i$, the aim is to handle, that is, to approximate, all paths of length at most $d / 2^i$ in the current residual graph $G[V_i]$.

However, we cannot stop after the first phase. In particular, after only \(\log d\) iterations, the residual subgraph \(G[V_{\log d}]\) may still be dense (we did not perform enough pruning rounds to force \(O(n^{1+1/k})\) total edges), so a second sparsification phase is required.

\subsection{The Construction - First Phase}
In this subsection we present the first phase of the algorithm which computes a subgraph $H \subseteq G$ with the following properties:
\begin{itemize}
  \item every original path of length at most $d$ in $G$ has been approximated in $H\cup G[V_{\log d}]$ with length $2k + O(d\log d)$, and
  \item the expected number of edges in $H$ is only $O\bigl((k + d\log d)\,n\bigr)$.
\end{itemize}

We now describe the procedure. It closely follows the high-level intuition, with three key adjustments: (1) we initialize the radius parameter to 
\(r_0 = k + 2\,d\log d,\) and (2) recurse via
\[
  r_{i+1} \;=\; \tfrac12\,r_i \;-\;\tfrac{d}{2^i},
\]
and (3) the algorithm runs exactly $\log d$ iterations. Now, in iteration $i$ our goal is to handle (i.e. approximate) paths of length at most $d/2^i$ in the current residual graph $G[V_i]$.

\textbf{First Phase.}
The algorithm begins by initializing $H = (V, \emptyset)$ and setting $r_0 = k + 2 d \log d$. For each iteration $0 \leq i < \log d$, starting with $i = 0$, the algorithm operates on the vertex set $V_i$ and proceeds as follows:
\begin{itemize}
    \item \textbf{Sample $C_i$:} It sets $C_i$ to be the set obtained by sampling each vertex in \(V_i\) independently with probability \(n^{-1/2^{i+1}}\).
    \item \textbf{Link centers:} For every two vertices $c,c' \in C_i$ whose distance is at most $r_i$ in $G[V_i]$, it adds a shortest path between them to $H$.
    \item \textbf{Link vertices to centers:} It performs a BFS from the centers \(C_i\) over the subgraph \(G[V_i]\) and adds the resulting BFS forest edges to $H$.
    \item \textbf{Prune:} It defines radius \(r_{i+1} = \tfrac{1}{2}r_i - \tfrac{d}{2^i}\), and removes all vertices \(x \in V_i\) with \(\dist_i(x, C_i) \leq r_{i+1}\). Then, it sets \(V_{i+1}\) to be the remaining vertex set.
\end{itemize}

Once all iterations complete, the algorithm returns \(H\) and \(V_\ld\). The pseudo-code is in Algorithm~\ref{al:spaOne}.

\SetAlgoNoEnd
\begin{algorithm}[H]
\label{al:spaOne}

Set $ H = ( V, \emptyset ) $, $V_0 = V$, and $ r_0 = k + 2 d \ld $.

\tcp{\textbf{First Phase}}

\For{$i=0$ \KwTo $\ld-1$}
{
\tcp{(1) Sample $C_i$\quad(2) Connect centers $\le r_i$\quad(3) BFS from $C_i$ \quad(4) Prune}
   Sample each vertex of $V_i$ independently with probability $n^{-1/2^{i+1}}$ to form $C_i$;\\
   For every $c,c'\in C_i$ with $\dist_i(c,c')\le r_i$, add a shortest path between them in $G[V_i]$ to $H$;\\
   Run a BFS from $C_i$ in $G[V_i]$ and add the resulting BFS forest to $H$;\\
   Set $r_{i+1}=\tfrac12 r_i-\tfrac{d}{2^i}$ and \(V_{i+1}=\{\,x\in V_i \mid \dist_i(x,C_i)>r_{i+1}\,\}\).

}

\Return $(H, V_{\log d})$

\caption{\spaOne receives a graph \( G \) and integers \( d, k \) such that \( d | k \) and \( d \) is a power of $2$. It returns a subgraph \( H \subseteq G\) with expected size \( O\big((k + d\log d)n\big) \), and a subset of the vertices $ V_\ld \subseteq V$.}
\end{algorithm}

\subsubsection{Analysis}

We start by analyzing the algorithm \spaOne. We begin by bounding the expected number of edges in $H$ at the end of the algorithm. To do so, we need the following claim.

\begin{lemma}
\label{cl:size_gen}
Let $ 0 \leq i < \log d $, and let $V_i$ be the vertex set at the beginning of iteration $i$. Then, for every vertex $ x \in V_i $, we can bound the expected size of the ball $B_{i+1}(x,r_{i+1})$:

\begin{equation*}
\mathbb{E}_{C_i} \left[ \lvert B_{i+1}(x, r_{i+1}) \rvert \right] < \frac{1}{e} n^{\sfrac{1}{2^{i+1}}},
\end{equation*}
where we set $ \lvert B_{i+1}(x, r_{i+1}) \rvert = 0 $ if $ x \notin V_{i+1} $. 

Consequently,
\[
\mathbb{E}_{C_{\ld-1}}\!\left[\,|B_{\ld}(x,r_\ld)|\,\right] < n^{1/2^{\ld}} = n^{1/d}.
\]
\end{lemma}

\begin{proof}
Let $ y \in B_i(x, r_{i+1}) $. If $ y \in C_i $, then the algorithm deleted $x$ during the iteration $i$, so $ \lvert B_{i+1}(x, r_{i+1}) \rvert = 0 $. The probability that $ y \notin C_i $ is exactly $ 1 - n^{-(1/2)^{i+1}} $, so we can bound the expectation of $ \lvert B_{i+1}(x, r_{i+1}) \rvert $:

\begin{align*}
\mathbb{E}_{C_i} \left[ \lvert B_{i+1}(x, r_{i+1}) \rvert \right] &\leq \left( 1 - n^{-(1/2)^{i+1}} \right)^{\lvert B_i(x, r_{i+1}) \rvert} \cdot \lvert B_i(x, r_{i+1}) \rvert \\
&\leq \lvert B_i(x, r_{i+1}) \rvert \cdot e^{- \lvert B_i(x, r_{i+1}) \rvert / n^{\sfrac{1}{2^{i+1}}}}.
\end{align*}

Now, by setting $ z = \lvert B_i(x, r_{i+1}) \rvert / n^{\sfrac{1}{2^{i+1}}} $, we can see that we get an equation of 
\begin{equation*}
f(z) = n^{\sfrac{1}{2^{i+1}}} \cdot z \cdot e^{-z}.
\end{equation*}

A standard derivative check shows that \(f\) is maximized at \(z=1\), giving \(f(1)=\tfrac1e\,n^{1/2^{i+1}}\).
By plugging it back, we get
\begin{equation*}
\mathbb{E}_{C_i} \left[ \lvert B_{i+1}(x, r_{i+1}) \rvert \right] < \frac{1}{e} n^{\sfrac{1}{2^{i+1}}}.
\end{equation*}

\end{proof}

Next, we analyze the expected size of the spanner $H$.

\begin{lemma}
\label{le:size1}
In expectation, the number of edges added to $H$ satisfies
\[
\sum_{i=0}^{\ld-1}\!\bigl(r_i n/(2e)+n\bigr)=O(r_0 n).
\]
\end{lemma}

\begin{proof} We show that in each iteration $i$, the expected number of edges added to $H$ is at most $r_i n / (2e) + n$.

In each iteration, edges are added from two steps: (1) \textbf{Link Centers} step, by adding shortest paths between every pair $c, c' \in C_i$ with $\dist_i(c, c') \leq r_i$, and (2) \textbf{Link vertices to centers} step, by performing a BFS from $C_i$ in $G[V_i]$, which contributes at most $n$ edges. Thus, it suffices to bound the expected number of edges added via the shortest paths between centers.

For any center $c \in C_i$, a path is added to another center $c' \in C_i$ only if $c' \in B_i(c, r_i)$. Hence, the number of paths for which $c$ is responsible is $|B_i(c, r_i) \cap C_i \setminus \{c\}|$, and the maximum length of a path is $r_i$. The expected contribution of such paths to the size of $H$ is at most
\begin{equation*}
    r_i \cdot \frac{1}{2} \sum_{c \in C_i} \e \left[ |B_i(c, r_i) \cap C_i \setminus \{ c \} | \right].
\end{equation*}

So it suffices to prove that
\begin{equation*}
\sum_{c \in C_i} \e[ |B_i(c, r_i) \cap C_i \setminus \{ c \} |] \leq n/e.
\end{equation*}

Since centers are chosen independently at random, for each $c \in C_i$, the expected number of other centers in $B_i(c, r_i)$ equals the expected ball size times the sampling probability:
\begin{equation*} \mathbb{E}_{C_i \setminus \{c\}} \left[ |B_i(c, r_i) \cap C_i \setminus \{ c \}| \right] = \mathbb{E} \left[ |B_i(c, r_i)| \right] \cdot n^{-\sfrac{1}{2^{i + 1}}}. \end{equation*}

Summing over all centers gives:
\begin{equation*} \begin{aligned} \sum_{c \in C_i} \mathbb{E} \left[ |B_i(c, r_i) \cap C_i \setminus \{ c \} | \right] & \leq \sum_{c \in C_i} \mathbb{E} \left[ |B_i(c, r_i)| \right] \cdot n^{-\sfrac{1}{2^{i + 1}}} & = \\ \mathbb{E}[|C_i|] \cdot \max_{c} \mathbb{E}[|B_i(c, r_i)|] \cdot n^{-\sfrac{1}{2^{i + 1}}} & \leq \frac{1}{e} \cdot n^{1 - \sfrac{1}{2^{i + 1}} + \sfrac{1}{2^i} - \sfrac{1}{2^{i + 1}}} = \frac{1}{e}n. \end{aligned} \end{equation*}

\end{proof}

We now turn to analyze the stretch. The first claim states that if some vertex $x \in V_i$ was removed during iteration $i$, all paths containing $x$ of length at most $d/2^i$ must have a good approximation in $H$:

\begin{lemma}
\label{cl:stretch_gen} Let $0 \leq i < \ld$ and let $x \in V_i \setminus V_{i+1}$ be a vertex removed during iteration $i$, meaning that $\dist_i(x, C_i) \leq r_{i+1}$. Then for every path $P \subseteq G[V_i]$ of length $|P| \leq \frac{d}{2^i}$ containing $x$, there exists a corresponding path $P' \subseteq H$ between the same endpoints such that
\begin{equation*}
|P'| \leq 4 \cdot r_{i+1} + 3 \cdot |P|,
\end{equation*}
implying that the stretch is at most
$
\frac{|P'|}{|P|} \leq 4 \cdot \frac{r_{i+1}}{|P|} + 3.
$
\end{lemma}

\begin{proof}
    Let $x \in V_i$ be a vertex such that $\dist_i(x, C_i) \leq r_{i+1}$, and let $P$ be a path between some vertices $s, t$ with length $|P| \leq \frac{d}{2^i}$ that passes through $x$. We define $c_i : V_i \to C_i$ as a function that assigns to each vertex the closest vertex in $C_i$ that was connected to it in the BFS forest in iteration $i$ in the \textbf{Link centers to vertices} step.

To prove the claim, we construct a path in $H$ between $s$ and $t$ that passes through their nearest centers, $c_i(s)$ and $c_i(t)$. We show that this path is contained in $H$ and that it satisfies the length bound stated in the lemma.

We start by proving that $\dist_i(c_i(s), c_i(t)) \leq r_i$, which ensures that $c_i(s)$ and $c_i(t)$ are connected in $H$ during the \textbf{Link centers} step. By the triangle inequality:
\[
\dist_i(c_i(s), c_i(t)) \leq \dist_i(c_i(s), s) + \dist_i(s, t) + \dist(t, c_i(t)).
\]

Next, we bound $\dist_i(c_i(s), s) + \dist(t, c_i(t))$ using the fact that $\dist_i(x, C_i) \leq r_{i+1}$:
\begin{equation}
\label{ineq1}
\begin{aligned}
\dist_i(c_i(s), s) + \dist(t, c_i(t)) &\leq \dist_i(c_i(x), s) + \dist_i(t, c_i(x)) \\
&\leq \dist_i(c_i(x), x) + \dist_i(x, s) + \dist_i(t, x) + \dist_i(x, c_i(x)) \\
&\leq 2r_{i+1} + |P|.
\end{aligned}
\end{equation}

Substituting back, we have:
\begin{equation}
\label{ineq2}
\dist_i(c_i(s), c_i(t)) \leq |P| + (2r_{i+1} + |P|) = 2r_{i+1} + 2|P| \leq r_i.
\end{equation}

Here, the last inequality holds because $ \frac{1}{2}r_i = r_{i+1} + \frac{d}{2^{i}} $. Thus, by construction $c_i(s)$ and $c_i(t)$ are connected in $H$ by a path of this length.
  
Now, we construct the path $P' \subseteq H$ and bound its length:
\[
\dist_H(s, t) \leq \dist_H(s, c_i(s)) + \dist_H(c_i(s), c_i(t)) + \dist_H(c_i(t), t).
\]
Combining inequalities (\ref{ineq1}), (\ref{ineq2}):
\[
\dist_H(s, t) \leq (2r_{i+1} + 2|P|) + (2r_{i+1} + |P|) = 4r_{i+1} + 3|P|.
\]
This completes the proof, as the path in $H$ satisfies the required length bound. See figure \ref{fig:stretch_gen}.
\end{proof}

\begin{figure}[h!]
    \centering
\begin{tikzpicture}[scale=1.5, every node/.style={font=\small}]

\coordinate (s) at (0,0);
\coordinate (x) at (2.5,0);
\coordinate (t) at (6,0);
\coordinate (cx) at (2.5,1.5);
\coordinate (cs) at (-1,2);
\coordinate (ct) at (7,2);

\draw[thick, black] (s) -- (cs) node[midway, left] {$\leq r_{i+1}+\dist_i(x,s)$};
\draw[thick, black] (x) -- (cx) node[midway, right] {$\leq r_{i+1}$};
\draw[thick, black] (cs) -- (ct) node[midway, above] {$\leq r_{i+1} + d / 2^i + (r_{i+1} + d / 2^i) = r_i$};
\draw[thick, black] (ct) -- (t) node[midway, right] {$\leq r_{i+1} +\dist_i(x,t)$};
\draw[thick, magenta, dashed] (s) -- (t) node[midway, below right, magenta] {$|P| \leq d / 2^i$};

\draw[thick, gray, dashed] (cx) -- (t) node[midway, above, black] {$\leq r_{i+1} + \dist_i(x,t)$};
\draw[thick, gray, dashed] (cx) -- (s) node[midway, above, black] {$\leq r_{i+1} + \dist_i(x,s)$};

\filldraw[black] (s) circle (2pt) node[below] {$s$};
\filldraw[black] (x) circle (2pt) node[below] {$x$};
\filldraw[black] (t) circle (2pt) node[below] {$t$};
\filldraw[black] (cs) circle (2pt) node[above] {$c_i(s)$};
\filldraw[black] (cx) 
circle (2pt) node[above] {$c_i(x)$};
\filldraw[black] (ct) circle (2pt) node[above] {$c_i(t)$};

\end{tikzpicture}
\caption{Let $P$ be a path between $s,t$ of length $\leq d/2^i$. Assuming that $\dist_i(x,c_i(x)) \leq r_{i+1}$, we can prove that $\dist_i(c_i(s),c_i(t)) \leq r_i$ and therefore $c_i(s),c_i(t)$ connected in $H$ by a path. So the path $P' = s \to c_i(s) \to c_i(t) \to t$ is entirely in $H$ and of length $ |P'| \leq 4r_{i+1} + 3 |P|$. The black paths are in $H$, and the dashed segments are in $G[V_i]$.}
\label{fig:stretch_gen}

\end{figure}
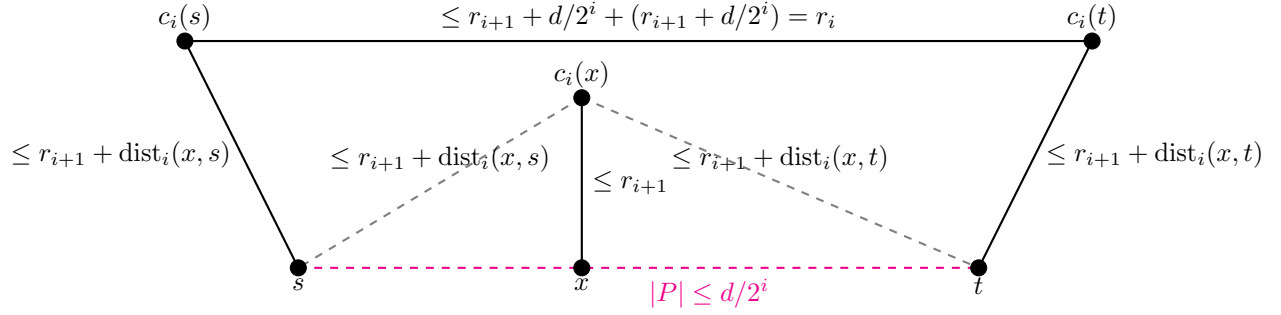

In the next claim, we provide the exact terms of $r_i$. We then substitute it into Lemma \ref{cl:stretch_gen} to get the actual approximation of the paths (without the dependency in $r_{i+1}$).

\begin{claim}
\label{cl:rsize_gen} 
Given $r_0 = k + 2d \log d$ and $r_{i+1} = \frac{1}{2} r_i - \frac{d}{2^i}$, then for every $i \geq 0$ we have:
\begin{equation*}
r_i =  \frac{1}{2^i} \left( k + 2d( \ld - i) \right).
\end{equation*}

In particular, when $i = \ld$, we get $r_{\ld} = \frac{k}{d}$, and when $i = \ld + 1$, we get $r_{\ld+1} = \frac{k}{2d} - 1$.
\end{claim}

\begin{proof}
We prove it by induction. We defined $ r_0 = k + 2d \log d$, so the equation holds. We assume that it holds for $r_i$ and prove it for $r_{i+1}$:
\begin{equation*}
    r_{i+1} = \frac{1}{2}r_{i} - \frac{d}{2^i} = \frac{1}{2} \cdot \frac{1}{2^i} \left( k + 2 d (\log d - i) \right) - \frac{2d}{2^{i+1}} = \frac{1}{2^{i+1}} \left( k + 2d \left( \log d - i - 1 \right) \right).
\end{equation*}    
\end{proof}

Combining Lemma~\ref{cl:stretch_gen} and Claim~\ref{cl:rsize_gen}, we can prove the main property about our subgraph $H$. Specifically, for all $0 \leq i \leq \ld$, any path $P$ of length at most $d/2^i$ in $G[V_i]$ has an alternative path $P'$ in $H \cup G[V_{\ld}]$ of length at most $4r_{i+1} + 3|P|$. In particular, setting $i = 0$ implies that for any path $P$ of length at most $d$ in $G = G[V_0]$, there exists a path $P' \subseteq H \cup G[V_{\ld}]$ such that \begin{equation*} |P'| \leq 4r_1 + 3|P| \le 2 \left( k + 2d( \ld - 1) \right) + 3d = 2k + 4d \ld + d. \end{equation*}

Ideally, this would complete our construction. However, \( G[V_{\ld}] \) might be too dense and could contain more than $O(n^{1+1/k})$ edges, so we cannot simply add \( \left| E(G[V_{\ld}]) \right| \) to $H$.
To overcome this, the goal of the second phase is to construct a subgraph \(H' \subseteq G[V_\ld]\) with \(O(n^{1+1/k})\) edges, such that \(H'\) is a \((4 \cdot r_{\ld + 1} + 3) = \left(\frac{2k}{d} - 1\right)\)-multiplicative spanner of \(G[V_\ld]\), where we use Claim~\ref{cl:rsize_gen}, which states that \(r_{\ld + 1} = \tfrac{k}{2d} - 1\). If such a subgraph exists, we can prove the following:
\begin{lemma} \label{cl:stretch_all_gen}
Let $0 \leq i \leq \log d$, and
    let $P$ be a path in $G[V_i]$ of length at most $d/2^i$. Additionally, let $H'$ be a $(\frac{2k}{d}-1)$-multiplicative spanner of $G[V_\ld]$.
    Then, in $H \cup H'$, there is an alternative path between the endpoints of $P$ of length at most $4r_{i+1} + 3|P|$.

    Consequently, by setting \( i = 0 \), we get that every path \( P \subseteq G \) of length at most \( d \) has an alternative path \( P' \subseteq H \cup H' \) between the same endpoints such that
    \[
    |P'| \leq 4r_1 + 3|P|
\leq 2 \left( k + 2d(\ld-1) \right) + 3d
= 2k + 4d\ld - d.
    \]
\end{lemma}

\begin{proof}
    We prove the claim by induction on \(i\), starting from \(i = \log d\) and proceeding down to \(i = 0\).

\medskip
\noindent
\textbf{Base case:} \(i = \log d\).  
In this case, the claim concerns every path \(P \subseteq G[V_\ld]\) of length $\tfrac{d}{2^{\ld}} = 1$—i.e., every edge in \(G[V_\ld]\). The claim states that such a path has stretch at most \((4r_{\ld+1} + 3)\) in \(H \cup H'\). Since \(H'\) is a \((\tfrac{2k}{d} - 1) = (4r_{\ld+1} + 3)\)-multiplicative spanner of \(G[V_\ld]\), this holds trivially. Here, we substituted \(r_{\log d + 1} = \tfrac{k}{2d} - 1\) as given by Claim~\ref{cl:rsize_gen}.

    \textbf{Induction step:} We assume this lemma is valid for $i+1$. Let $P \subseteq G[V_i]$ be a path of length $d/2^i$ between $s,t$. Note that if $P$ is not fully contained in $G[V_{i+1}]$, is means that there is some vertex $x \in P$ that the algorithm deleted during iteration $i$, meaning that $x \in V_i \setminus V_{i+1}$, and by Lemma \ref{cl:stretch_gen} the lemma holds.
    
    So we can assume that $P \subseteq G[V_{i+1}]$. We can now split $P$ in the middle into two paths, $P = P_1\cup P_2 $, where each of them is of length at most $d/2^{i+1}$, and apply the induction step to get the desired stretch. Let $P'_1, P'_2$ be the alternative paths of $P_1,P_2$ in $H$, so $P' = P'_1 \cup P'_2$ is the alternative path of $P$ in $H$ and of length at most
    \begin{equation*}
        |P'| \leq |P'_1| + |P'_2| \leq 3|P_1| + 4r_{i+2} + 3|P_2| + 4r_{i+2} = 3(|P_1| + |P_2|) + 8r_{i+2}< 3|P| + 4r_{i+1},
    \end{equation*}
    as $r_{i+2} < \frac{1}{2}r_{i+1}$ for every $i$.
\end{proof}

We now turn to the second phase, whose goal is to construct the spanner \(H'\) of \(G[V_\ld]\).

\subsection{The Construction – Second Phase}
In this subsection, we construct a \((\tfrac{2k}{d} - 1)\)-spanner \(H' \subseteq G[V_\ld]\) with \(O(n^{1 + 1/k})\) edges in expectation. We begin by outlining the intuition behind the construction, followed by the formal algorithm and its analysis.

\paragraph{Intuition} Assume for now that Claim \ref{cl:size_gen} holds for every $x \in V_\ld$, not only in expectation, i.e, that for every $x \in V_\ld$,
\begin{equation}
    \label{eq:ball_gen}
|B_\ld(x, r_\ld)| = |B_\ld(x, \tfrac{k}{d})| < n^{1/d}.
\end{equation}
 We claim that under this assumption, it is possible to construct a $(\frac{2k}{d} - 1)$-multiplicative spanner of $G[V_\ld]$ using a known approach. This part of the construction closely resembles the "classic" algorithm for building a $(2k - 1)$-multiplicative spanner, as presented in \cite{peleg1988graph}.

The algorithm proceeds as follows: Choose an arbitrary vertex $x \in V_\ld$, and grow a ball centered at $x$ by gradually increasing the radius $i$, stopping at the first radius for which
\begin{equation*}
    |B_{\ld}(x, i )| \leq n^{1/k} |B_{\ld}(x, i-1)|.
\end{equation*}
By Inequality~\ref{eq:ball_gen}, and assuming it holds for every vertex (rather than only in expectation), we are guaranteed that this stopping condition occurs at \(i \leq \tfrac{k}{d}\).

Once such an $i$ is found, the algorithm adds to $H'$ a BFS tree rooted at $x$ spanning all vertices within distance $i$, and then removes from $V_\ld$ all vertices in the tree up to depth $i - 1$. This process repeats until $V_\ld$ is exhausted. Finally, the algorithm returns the spanner $H'$.

As in the classical setting, it is straightforward to verify that $H'$ is a $(\frac{2k}{d} - 1)$-multiplicative spanner for $G[V_\ld]$, and that it contains \(O(n^{1+1/k})\) edges.

Unfortunately, inequality \ref{eq:ball_gen} holds only in expectation, so we cannot directly rely on it to guarantee that the radius $i$ is always at most $\frac{k}{d}$. Instead, we adapt the idea by designing a similar approach that carefully controls the expected number of edges added to $H'$. The goal is to ensure that $H'$ remains sparse in expectation while still providing a good stretch guarantee.

Before proceeding to the Second Phase, we introduce a key sparsity notion:

\begin{definition}[\((R,k)\)-Sparsity]
  Let \(\mathcal{G}\) be a distribution over graphs on a fixed vertex set \(V\) with $n$ vertices, and let \(R, k\) be integers satisfying \(R \leq k\).  We say \(\mathcal{G}\) is \emph{\((R,k)\)-sparse} if for every \(x\in V\),
\[
  \e_{G_s\sim\mathcal{G}}\!\bigl[\lvert B_{G_s}(x,R)\rvert\bigr]
  \;\le\;
  n^{R/k}.
\]
\end{definition}

In our setting we take \(R=\tfrac{k}{d}\). By Lemma~\ref{cl:size_gen}, for every \(x\) we have \(\e[|B_{\ld}(x,R)|]\le n^{R/k}\), so \(G[V_\ld]\) is \((R,k)\)-sparse by definition.

\paragraph{Algorithm Description}

In the second phase, we are given an \((R, k)\)-sparse graph \(G_s = (V_s, E_s)\), and our goal is to construct a \((2R - 1)\)-multiplicative spanner \(H_s \subseteq G_s\) with only \(O(n^{1 + 1/k})\) edges in expectation.

We denote by \(B_s(v, r) := B_{G_s}(v, r)\) the ball of radius \(r\) around vertex \(v \in V_s\) in the graph \(G_s\).

Before presenting the algorithm, we outline the core intuition behind it. At each iteration, the goal is to find a relatively \emph{sparse} vertex \texttt{center} and a radius \(i \leq R\) such that
\[
  |B_s(\texttt{center}, i)| \;\le\; n^{1/k} \cdot |B_s(\texttt{center}, i-1)|,
\]
so that adding a BFS tree of depth \(i\) rooted at \texttt{center} to \(H_s\) and removing the inner ball \(B_s(\texttt{center}, i-1)\) from \(V_s\) incurs a cost of at most \(O(n^{1/k})\) edges per removed vertex.

However, such a sparse vertex may not always exist. In that case, the algorithm allows proceeding with a \emph{dense} center \texttt{center} for which
\[
  |B_s(\texttt{center}, R)| > n^{1/k} \cdot |B_s(\texttt{center}, R-1)|, \quad \text{and} \quad |B_s(\texttt{center}, R-1)| > n^{(R-1)/k}.
\]
In this case, the algorithm adds a BFS tree of depth \(R\) rooted at \texttt{center} to \(H_s\), and removes the inner ball \(B_s(\texttt{center}, R-1)\) from \(V_s\). Note that every vertex in \(B_s(\texttt{center}, R-1)\) is itself dense with respect to radius \(R\), so that while this step may incur a higher cost per edge - more than \(O(n^{1/k})\) edges added per removed vertex — it is applied only to a small fraction of the vertices. This is justified by the \((R,k)\)-sparsity of \(G_s\), which ensures that the expected number of such dense vertices is small. Thus, although some steps may be expensive, they are rare; see Lemma~\ref{cl:size2}, which shows the total expected edges are \(O(n^{1+1/k})\).

To bound the cost of such costly steps, the algorithm always selects the vertex with the smallest estimated density as the next \texttt{center}.

The algorithm proceeds as follows. It initializes the spanner \(H_s = (V_s, \emptyset)\). Then, for each vertex \(x \in V_s\), it assigns a \emph{budget} defined by:
\[
b(x) \;=\; \max\!\left\{\,n^{1/k},\;\frac{\lvert B_s(x,R)\rvert}{n^{(R - 1)/k}}\,\right\},
\]
where \(B_s(x, R)\) is the ball of radius \(R\) around \(x\) in \(G_s\).

The algorithm then iteratively builds the spanner while pruning vertices from \(V_s\). While \(V_s\) is non-empty:

\begin{itemize}
    \item It selects a vertex \(\texttt{center} \in V_s\) with the minimum budget value, i.e.,
    \[
      \texttt{center} = \mathop{\arg\min}_{x \in V_s} \; b(x).
    \]
    \item For \(i = 1\) to \(R\), do:
    \begin{itemize}
        \item If either
        \[
        |B_s(\texttt{center}, i)| \leq n^{1/k} \cdot |B_s(\texttt{center}, i-1)| \quad \text{or} \quad i = R,
        \]
        then:
        \begin{itemize}
            \item It adds the depth-\(i\) BFS tree rooted at \texttt{center} to \(H_s\).
            \item Removes all vertices in \(B_s(\texttt{center}, i-1)\) from \(V_s\).
            \item \textbf{Breaks} and restarts from the outer loop.
        \end{itemize}
    \end{itemize}
\end{itemize}

When \(V_s\) becomes empty, Algorithm~\ref{al:spaTwo} terminates and returns the spanner \(H_s\). Each extracted BFS tree has radius at most \(R\), so the stretch is bounded by \(2R-1\) (Lemma~\ref{cl:2Rstretch}).

The pseudo-code of \spaTwo is apparent in Algorithm \ref{al:spaTwo}.

\SetAlgoNoEnd
\begin{algorithm}[H]
\label{al:spaTwo}

\tcp{\textbf{Second Phase}}

Set $ H_s = ( V_s, \emptyset ) $.

For each $x \in V_s$, compute
\(
b(x) \;=\; \max\!\left\{\,n^{1/k},\;\frac{\lvert B_s(x,R)\rvert}{n^{(R - 1)/k}}\,\right\}.
\)

\While{$V_s$ is not empty }{

Set $\texttt{center} = \arg\min_{x \in V_s}\{ b(x) \}$.

\For{$i=1$ \KwTo $R$}{
  \If{$|B_s(\texttt{center},i)| \le n^{1/k}\,|B_s(\texttt{center},i-1)|$ \textbf{ or } $i=R$}{
    add the depth-$i$ BFS tree rooted at \texttt{center} to $H_s$;\\
    remove all vertices in $B_s(\texttt{center},i-1)$ from $V_s$;\\
    \textbf{break};
  }
}

}
\Return $H_s$

\caption{\spaTwo receives a $(R,k)$-sparse graph $G_s=(V_s,E_s)$ and returns a $(2R-1)$-multiplicative spanner $H$ with $O(n^{1+1/k})$ edges in expectation.}
\end{algorithm}

\subsubsection{Analysis}

We start analyzing Algorithm \spaTwo. We first show that $H_s$ provides the desired stretch, and then we bound its number of edges. The first claim states that $H$ is indeed a $(2R - 1)$-multiplicative spanner of $G_s$.

\begin{lemma}
\label{cl:2Rstretch}
$H_s$ is a $(2R - 1)$-multiplicative spanner of $G_s$.
\end{lemma}

\begin{proof}

Consider an edge $e = (u, v)$ that is removed from $G_s$. Without loss of generality, we assume that $v$ is removed before $u$. At that stage, for some $i \leq R$, there is a vertex $\texttt{center} \in G_s$ such that $v \in B_s(\texttt{center}, i-1)$, and a BFS tree of radius $i$ is added to $H_s$. Since $u$ is also within this BFS tree, there exists a path in $H_s$ between $v$ and $u$ that passes through $\texttt{center}$, with length at most $i - 1 + i = 2i - 1 \leq 2R - 1$.

\end{proof}

We proceed with claiming that the algorithm didn't add too many edges to $H_s$.

\begin{lemma} \label{cl:size2}
At the end of the algorithm \spaTwo, the number of edges in \(H_s\) is \(O(n^{1+1/k})\) in expectation.
\end{lemma}

\begin{proof}
We first show that if each vertex \(x \in V_s\) is charged at most \(b(x)\) edges in expectation, then the total number of edges added to \(H_s\) satisfies
\[
\mathbb{E}\left[\sum_{x \in V_s} b(x)\right] = O(n^{1+1/k}).
\]

Indeed, by linearity of expectation and the definition of the budget function,
\begin{align*}
\mathbb{E}_{G_s \sim \mathcal{G}}\left[ \sum_{x \in V_s} b(x) \right]
&= \sum_{x \in V} \mathbb{E}\left[ b(x) \right] \\
&\leq \sum_{x \in V} \left( n^{1/k} + \frac{\mathbb{E}[|B_s(x, R)|]}{n^{(R-1)/k}} \right) \\
&= n^{1 + 1/k} + \frac{1}{n^{(R-1)/k}} \sum_{x \in V} \mathbb{E}\left[ |B_s(x, R)| \right] \\
&\leq n^{1+1/k} + \frac{n^{1 + R/k}}{n^{(R-1)/k}} = 2n^{1+1/k},
\end{align*}
where the last inequality follows from the \((R,k)\)-sparsity assumption:
\[
\sum_{x \in V} \mathbb{E}[|B_s(x, R)|] \leq n^{1 + R/k}.
\]

Next, we argue that in every removal step, each vertex $x$ is charged at most \(b(x)\) edges:

\begin{enumerate}
    \item \textbf{Case 1:} If for some radius \(i \leq R\),
    \[
    |B_s(\texttt{center}, i)| \leq n^{1/k} \cdot |B_s(\texttt{center}, i-1)|,
    \]
    then each vertex \(x \in B_s(\texttt{center}, i-1)\) pays
    \[
    \frac{|B_s(\texttt{center}, i)|}{|B_s(\texttt{center}, i-1)|} \leq n^{1/k} \leq b(x).
    \]

    \item \textbf{Case 2:} If the above inequality fails for all \(i < R\), the algorithm adds at most \(|B_s(\texttt{center}, R)|\) edges and removes all vertices in \(B_s(\texttt{center}, R-1)\). 

    In particular, since the stopping condition did not hold for any \(i < R\), we must have
    \[
        |B_s(\texttt{center}, R-1)| > n^{(R-1)/k}.
    \]

    Since \(\texttt{center}\) was chosen as the vertex minimizing the budget over the set of remaining vertices, we can evenly distribute the cost among the removed vertices. Then each vertex \(x \in B_s(\texttt{center}, R-1)\) pays,
    \[
     \frac{|B_s(\texttt{center}, R)|}{|B_s(\texttt{center}, R-1)|} \leq \frac{|B_s(\texttt{center}, R)|}{n^{(R-1)/k}} \leq b(\texttt{center}) \leq b(x).
    \]

\end{enumerate}

Combining both cases, we conclude that the total number of edges added to \(H_s\) is at most \(\sum_{x \in V_s} b(x)\), which is \(O(n^{1+1/k})\) in expectation.
\end{proof}

\subsection{Combining All Together}

We now combine the two phases described above into a full construction of our spanner. The algorithm works as follows. It takes as input a graph \(G=(V,E)\) and integers \(d\le k\) with \(d\) a power of two and \(d\mid k\). It then
\begin{itemize}
    \item[$\bullet$] Sets $(H, V_\ld) =$ \spaOne.
    \item[$\bullet$] Sets $H_s =$ \mbox{\sf SpannerPartTwo(G[V}$_{\ld}$\mbox{\sf],k/d,k)}.

    \item[$\bullet$] Finally, it returns the union \(H \cup H_s \).
\end{itemize}

We now state the properties of the final spanner:

\begin{lemma}
The final spanner \(H \cup H_s\) satisfies the following:
\begin{itemize}
    \item[1.] For every pair of vertices \(u, v \in V\) with \(\dist_G(u,v) \leq d\), we have
    \[
    \dist_{H \cup H_s}(u,v) \leq 2k + 4d\ld + d.
    \]
    \item[2.] The expected number of edges in \(H \cup H_s\) is
    \(
    O(n^{1+1/k} + (k + d \log d) \cdot n).
    \)
\end{itemize}
\end{lemma}

\begin{proof}
Property 1 follows directly from Lemma~\ref{cl:stretch_all_gen}, by setting \(H' = H_s\), and $i=0$, since we proved in Lemma~\ref{cl:2Rstretch} that \(H_s\) is a \((\tfrac{2k}{d} - 1)\)-multiplicative spanner for \(G[V_\ld]\).

Property 2 follows by combining Lemma~\ref{le:size1}, which bounds \(|E(H)| = O((k + d \log d) \cdot n)\) in expectation, with Lemma~\ref{cl:size2}, which shows that \(|E(H_s)| = O(n^{1+1/k})\) in expectation.
\end{proof}

\subsection{General \(d\) and \(k\)}

We assumed until now that \(d\) is a power of 2 and that \(d|k\) (i.e., \(d\) divides \(k\)). If that is not the case, let \(\hat{d} \in [d, 2d)\) be the smallest power of 2 that is at least \(d\), and let \(\hat{k} \in [k, k + \hat{d})\) be the smallest integer such that \(\hat{d}\) divides \(\hat{k}\). Running the algorithm with \(\hat{d}, \hat{k}\) ensures that the size of the resulting spanner \(H \cup H_s\) is
\[
O\left(n^{1 + 1/\hat{k}} + (\hat{k} + \hat{d} \log \hat{d})n\right) = O\left(n^{1 + 1/k} + (k + d \log d)n\right).
\]
By Lemma \ref{cl:stretch_all_gen}, for every pair of vertices at distance at most \(d \leq \hat{d}\) from each other in \(G\), their distance in \(H \cup H_s\) is at most
\[
2\hat{k} + 4\hat{d} \log \hat{d} + \hat{d} \leq 2(k + 2d) + 8d(\log d + 1) + 2d = 2k + 8d \log d + 14d.
\]

\section{The $k$-Hybrid Spanner Construction} \label{sec:h-spa}

In this section, we present a $k$-hybrid spanner construction that guarantees $O(n^{1 + 1/k} + kn)$ expected edges, which is an improvement over the previous construction by Parter \cite{Bypassing}, which had $O(k^2 n^{1 + 1/k})$ edges. For every fixed $k$, this edge bound is optimal up to a constant factor under the girth conjecture.

This algorithm builds on the Baswana–Sen \((2k-1)\)-spanner and the \((\tfrac{3}{2}k,k-1)\)-spanner of \cite{Baswana_linear_weighted}, following the formulation in \cite{alpha_beta_05}. We first recall that one can construct a \((2k-1)\)-spanner with \(O(n^{1+1/k}+kn)\) edges in expectation in \(O(km)\) time.

\paragraph{Definitions and Notation.}
Let \(G = (V, E)\) be an unweighted graph. For vertex subsets \(C, C' \subseteq V\), define
\[
  E(C, C') \;=\; (C \times C') \cap E(G),
\]
the set of edges with one endpoint in \(C\) and the other in \(C'\). If \(C = \{v\}\), we write \(E(v, C')\) or \(E(C', v)\).

A \emph{cluster} is any subset \(C \subseteq V\), and a \emph{clustering} \(\C\) is a collection of pairwise vertex-disjoint clusters. A vertex \(v\) is said to be \emph{adjacent} to a cluster \(C\) if there exists \(u \in C\) such that \((v, u) \in E\).

\medskip
Our construction builds a sequence of \(k+1\) clusterings
\[
  \C_0, \C_1, \dots, \C_k,
  \quad
  \C_0 = \bigl\{\{v\} : v \in V\bigr\},
  \quad
  \C_k = \emptyset,
\]
satisfying \(\lvert \C_i \rvert \le n^{1 - i/k}\) for all \(0 \le i \le k\) in expectation. A vertex \(u\) is called \emph{\(i\)-clustered} if it belongs to some cluster \(C \in \C_i\), and \emph{\(i\)-unclustered} otherwise. For an \(i\)-clustered vertex \(u\), let \(C_i(u)\) denote the unique cluster in \(\C_i\) containing it.

We begin by building the Baswana \etal\ (\cite{alpha_beta_05}) \((2k-1)\)-spanner \(H \subseteq G\) of size $O(n^{1+1/k}+kn)$, by adding edges according to two rules:

\begin{description}
  \item[Rule R1 (cluster-trees).]
    For each \(C\in\C_i\), include in \(H\) a cluster tree of \(C\) of radius at most \(i\) (with respect to its chosen root).

  \item[Rule R2 (cover unclustered vertices).]
    The \textbf{first time} a vertex \(v\) becomes unclustered at level \(i\) (i.e., \(v \in V(\C_{i-1})\) but \(v \notin V(\C_i)\)), include in \(H\) exactly one edge from each set
    \[
      E(v, C'), \quad \text{for every cluster } C' \in \C_{i-1} \text{ adjacent to } v.
    \]
\end{description}

Rule R1 is identical to that in \cite{alpha_beta_05}. Our Rule R2 differs slightly: we apply it only the first time a vertex becomes unclustered, instead of potentially at every level. As a result, the number of edges added is no larger, while preserving correctness. 

\medskip
\paragraph{Randomized Clustering Construction}
The randomized construction of Baswana and Sen \cite{Baswana_linear_weighted} builds each clustering $\C_{i+1}$ by independently sampling every cluster in $\C_i$ with probability $n^{-1/k}$ (and setting $\C_k=\emptyset$).  Starting from 
\[
  H \;=\; (V, \emptyset),\qquad \C_0 \;=\;\bigl\{\{v\}\mid v\in V(G)\bigr\},
\]
the algorithm iterates for $i=1,\dots,k$ as follows:
\begin{enumerate}
  \item Sample each cluster in $\C_{i-1}$ with probability~$n^{-1/k}$ to form $\C_i$ (if $i=k$, set $\C_k=\emptyset$).
  \item In parallel, for each vertex $v\in V(G)$:
    \begin{itemize}
      \item[(R1)] If $v$ does not already belong to a sampled cluster and has a neighbor in some sampled cluster $C\in \C_i$, add $v$ to that cluster and include an arbitrary edge from $E(v,C)$ in $H$.
      \item[(R2)] As in Rule R2 above.
    \end{itemize}
\end{enumerate}

We state the resulting guarantee without proof:

\begin{theorem}[Restatement of Lemma~4.1 and a remark in~\cite{alpha_beta_05}]
\label{th:2k-1-spanner}
For any integer \(k \ge 1\), one can construct a \((2k - 1)\)-spanner of \(G\) with \(O(n^{1+1/k} + kn)\) edges in \(O(km)\) time in expectation.
\end{theorem}

\subsection{Algorithm Description}

Our algorithm modifies the \((2k-1)\)-spanner framework to obtain a \(k\)\nobreakdash-hybrid spanner.  We highlight the three main changes:

\medskip
\noindent\textbf{1. Truncated clustering rounds.}  
Instead of running all $k$ levels, stop clustering after $i=\bigl\lceil k/2\bigr\rceil$. We note that Baswana-Sen had already observed (in the unweighted setting) that one can truncate after about $k/2$ levels and then interconnect the remaining high-level clusters; this is precisely the mechanism behind their $(\frac{3}{2}k,k\!-\!1)$–spanner construction.

\noindent\textbf{2. Modified edge‑cover rule (R2*).}  
Rule R2 originally applied only when a vertex first becomes unclustered. Now, we apply it both
\begin{itemize}
  \item at the first level \(i\) where a vertex \(v\) becomes unclustered, and
\item again at the final level \(i=\lceil \tfrac{k}{2}\rceil\), if \(v\) is \(i\)-unclustered,
\end{itemize}
by including in \(H\) exactly one arbitrary edge from each set
\[
  E\bigl(v,C'\bigr),
  \quad
  C'\in\C_{i-1}\text{ adjacent to }v.
\]

\medskip
\noindent\textbf{3. Inter‑cluster connections.}  
To achieve the hybrid stretch, we add two new rules that connect clusters across the midpoint:

\begin{description}
    \item[Rule R3.]  
    For every cluster \(C\in\C_{\lfloor k/2\rfloor - 1}\) and every cluster \(C'\in\C_{\lceil k/2\rceil}\) whose \emph{distance} in \(G\) is at most $2$, include up to two edges in \(H\) that realize a shortest path of length \(\le 2\) between \(C\) and \(C'\).
    
  \item[Rule R4.]  
    For every cluster \(C\in\C_{\lfloor k/2\rfloor}\) and every cluster \(C'\in\C_{\lceil k/2\rceil-1}\) that are adjacent, with at least one edge between them in \(G\), include one such edge in \(H\).
\end{description}

\medskip
Putting everything together:

\begin{enumerate}
  \item[] \textbf{Rules R1 and R2*.} Start to compute a randomized $(2k-1)$-spanner $H$ and stop at level $i = \lceil \tfrac{k}{2} \rceil$.
  \item[] \textbf{Rule 3.} Add to $H$ at most two edges between every pair of clusters \(C\in\C_{\lfloor k/2\rfloor - 1}\) and \(C'\in\C_{\lceil k/2\rceil}\) whose distance in \(G\) is at most $2$.
  
  \item[] \textbf{Rule 4.} Add to $H$ an edge between every pair of adjacent clusters \(C\in\C_{\lfloor k/2\rfloor}\) and \(C'\in\C_{\lceil k/2\rceil-1}\).
\end{enumerate}

We analyze, in expectation, the number of edges added to \(H\) by each rule:

\begin{itemize}
    \item[] \textbf{Rule R1.} For each clustering level \(i\), a spanning tree is added for every cluster in \(\C_i\), with total size \(O(n)\). Across \(\lceil k/2\rceil\) levels, Rule R1 contributes \(O(\lceil k/2\rceil\cdot n)=O(kn)\) edges.

    \item[] \textbf{Rule R2*.} Rule R2* is applied to each vertex at most twice, and each application contributes \(O(n^{1/k})\) edges in expectation. The total number of edges added by this rule is therefore \(O(n^{1+1/k})\).

    \item[] \textbf{Rule R3.} For every pair of clusters \(C \in \C_{\lfloor k/2\rfloor - 1}\) and \(C' \in \C_{\lceil k/2\rceil}\) with distance at most $2$ in \(G\), we add up to two edges. The expected number of such pairs is bounded by
    \[
\e[ \#\text{pairs} ] \le 2\cdot \e[|\C_{\lfloor k/2\rfloor-1}|]\cdot \e[|\C_{\lceil k/2\rceil}|]
= O\!\big(n^{1 - (\lfloor k/2\rfloor - 1)/k}\cdot n^{1 - \lceil k/2\rceil/k}\big)
= O(n^{1+1/k}).
\]

    \item[] \textbf{Rule R4.} For every adjacent pair of clusters
\(C\in\C_{\lfloor k/2\rfloor}\) and
\(C'\in\C_{\lceil k/2\rceil-1}\), we include one edge.
The expected number of such pairs is at most
\[
\mathbb{E}\!\left[
|\C_{\lfloor k/2\rfloor}|\,
|\C_{\lceil k/2\rceil-1}|
\right]
= O\!\left(
n^{1-\lfloor k/2\rfloor/k}
n^{1-(\lceil k/2\rceil-1)/k}
\right)
= O(n^{1+1/k}).
\]
\end{itemize}

\noindent
Summing over all rules, the expected number of edges in the spanner \(H\) is
\(
  O(n^{1 + 1/k} + kn),
\)
as claimed.

As for implementation, Rule R4 can be applied in linear time for any fixed level index \(i\), using standard BFS and cluster bookkeeping. In contrast, Rule R3 may require up to \(O(m\sqrt{n})\) time in total: for each cluster \(C \in \C_{\lceil k/2\rceil}\), we perform a BFS to depth $2$ and identify intersecting clusters from \(\C_{\lfloor k/2\rfloor-1}\); each such BFS may scan up to \(O(m)\) edges in total. Since \( \e \left[ |\C_{\lceil k/2\rceil}| \right] = O(n^{1/2})\), the overall time for Rule R3 is bounded by \(O(m\sqrt{n})\) in expectation.

After establishing correctness of Rules R1–R4, we obtain:

\begin{theorem}
    For any undirected unweighted graph $G$ and an integer \(k \ge 1\), one can construct a $k$-hybrid spanner of \(G\) with \(O(n^{1+1/k} + kn)\) edges in \(O(m \sqrt{n})\) time in expectation.
\end{theorem}

It remains to show that \(H\) is a \(k\)-hybrid spanner. First, we state a lemma about unclustered vertices:

\begin{lemma}
\label{le:hstretch}

Let \(v\in V\) be unclustered in some level \(l\le \lceil k/2\rceil\). Then for any edge \((v,u)\in E\), we have \(\dist_H(v,u)\le k\).

\end{lemma} 

\begin{proof}
Let $(u, v)$ be an arbitrary edge in the
original graph. Let $l\leq \lceil \tfrac{k}{2}\rceil$ be minimum
such that either $u$ or $v$ was $l$-unclustered, and without loss of generality suppose it is $v$. By Rule R2* there must be
an edge in $H$ from $v$ to $C_{l-1}(u)$ — call it $(v, w)$ —
and by Rule R1 there must be a path in $H$ from
$u$ to $w$ of length at most $2(l-1)$, twice the radius
of $C_{l-1}(u)$. Since $l \leq \lceil \tfrac{k}{2}\rceil $ it follows immediately that
\[
\dist_H(v,u) \leq \dist_H(v,w) + \dist_H(w,u) \leq 2l - 1 \leq k.
\]
\end{proof}

\begin{lemma}
\label{lem:k-hybrid-spanner}
The graph \(H\) is a \(k\)\nobreakdash-hybrid spanner of \(G\).
\end{lemma}

\begin{proof}
We analyze the stretch separately for all pairs of vertices at distances $1,2$ and $3$, and then extend the result to arbitrary distances using concatenation. Specifically, we prove:
\begin{enumerate}
  \item \textbf{Stretch for single edges.}  
    Let \((v,u)\in E(G)\).  
    \begin{enumerate}
      \item If either \(v\) or \(u\) becomes unclustered in some round, then by Lemma~\ref{le:hstretch} we have
        \(\dist_H(v,u)\le k \le 2k-1\).
      \item Otherwise both survive through level \(\lfloor k/2\rfloor\).  Then 
        \(C_{\lfloor k/2\rfloor-1}(v)\) and \(C_{\lceil k/2\rceil}(u)\) 
        are adjacent in \(G\), so by Rule R3 there is an edge \((x,y)\in H\) with 
        \(x\in C_{\lfloor k/2\rfloor-1}(v)\), \(y\in C_{\lceil k/2\rceil}(u)\).  
        By Rule R1 there is a path of length \(\le2\lfloor k/2\rfloor-2\) from \(v\) to \(x\) and length \(\le2\lceil k/2\rceil\) from \(y\) to \(u\).  Hence
        \[
          \dist_H(v,u)
          \;\le\;
          2\lfloor\tfrac{k}{2}\rfloor - 2
          \;+\;1\;+\;2\lceil\tfrac{k}{2}\rceil
          \;=\;
          2k-1.
        \]
    \end{enumerate}

  \item \textbf{Stretch for paths of length 2.}  
    Let \(P=(v,w,u)\) be a shortest path of length \(2\).  We consider three cases:
    \begin{enumerate}
\item[\textbf{Case 1}]

If $w$ was not clustered at some point, or if both $u$ and $v$ were not clustered at some point (can be in different iterations), by applying Lemma \ref{le:hstretch}, we get that in $H$, the distances between $ v $ and $w$ and between $w$ and $u$ are at most $k$, so $ \dist_H(v,u) \leq 2k $.

\end{enumerate}

In the next cases, we can assume that $w$ was clustered all the way, and also one of $v, u$. We assume without loss of generality that it is $v$. We consider two cases, depending on whether $u$ was $\lceil \frac{k}{2} \rceil $-clustered. See figure \ref{fig:two-cases}.

\begin{enumerate} 

\item[\textbf{Case 2}]

If \(u\) is \(\lceil \tfrac{k}{2} \rceil\)-clustered, then \(\dist_G\!\big(C_{\lfloor k/2\rfloor-1}(v),\, C_{\lceil k/2\rceil}(u)\big)\le 2\), and by Rule R3 the algorithm connected $ C_{\lfloor \frac{k}{2} \rfloor - 1}(v) $ and $ C_{\lceil \frac{k}{2} \rceil }(u) $ by a path of length at most two, so the distance between $u$ and $v$ is at most
\begin{equation*}
\dist_H(v,u) \leq 2 (\lfloor \frac{k}{2} \rfloor - 1) + 2 + 2 ( \lceil \frac{k}{2} \rceil ) = 2( \lfloor \frac{k}{2} \rfloor +  \lceil \frac{k}{2} \rceil) = 2k.  
\end{equation*}

\item[\textbf{Case 3}] 

If \(u\) is \(\lceil \tfrac{k}{2} \rceil\)-unclustered, then by Rule R2*, the algorithm adds an edge from \(u\) to every cluster
\(C \in \C_{\lceil \tfrac{k}{2} \rceil - 1}\) adjacent to \(u\). Since \(w\) is
\((\lceil \tfrac{k}{2} \rceil - 1)\)-clustered and \((u,w)\in E\), there is an edge
\((u,t)\in H\) with \(t \in C_{\lceil \tfrac{k}{2} \rceil - 1}(w)\).
By Rule R4, there is also an edge \((x,y)\in H\) with
\(x \in C_{\lfloor \tfrac{k}{2} \rfloor}(v)\) and \(y \in C_{\lceil \tfrac{k}{2} \rceil - 1}(w)\).
Therefore,
\begin{align*}
\dist_H(v,u)
&\le \dist_H(v,x)
   + \underbrace{\dist_H(x,y)}_{=1}
   + \dist_H(y,t)
   + \underbrace{\dist_H(t,u)}_{=1} \\[4pt]
&\le 2\lfloor \tfrac{k}{2}\rfloor + 1
   + 2(\lceil \tfrac{k}{2}\rceil - 1) + 1 = 2k.
\end{align*}

\end{enumerate}

\begin{figure}[h!]
  \centering
  \begin{minipage}{0.45\textwidth}
    \centering
    \begin{tikzpicture}[every node/.style={font=\small}]

  \node[draw, circle, minimum size=2cm] (Cv) at (0,0) {$C_{\lfloor k/2\rfloor - 1}(v)$};
  \node[draw, circle, minimum size=2.3cm] (Cu) at (3,0) {$C_{\lceil k/2\rceil}(u)$};
  \coordinate (v) at (0.8,-0.7);
  \coordinate (u) at (2.1,-0.7);
  \coordinate (x) at (0.8, 0.7);
  \coordinate (y) at (2.1, 0.7);
  \coordinate (c1) at (0, 0);
  \coordinate (c2) at (3, 0);

  \filldraw[black] (v) circle (2pt) node[below] {$v$};
  \filldraw[black] (u) circle (2pt) node[below] {$u$};
  \draw[thick, gray, dotted] (v) -- (c1);
  \draw[thick, gray, dotted] (x) -- (c1);
    \draw[thick, gray, dotted] (y) -- (c2);
  \draw[thick, gray, dotted] (u) -- (c2);
  \draw[thick, magenta, dashed] (u) -- (v) node[midway, below, magenta] {$2$};
  \draw[thick, black] (x) -- (y) node[midway, above, black] {$2$};

\end{tikzpicture}
  \end{minipage}
  \hfill
  \begin{minipage}{0.45\textwidth}
    \centering
    \begin{tikzpicture}[every node/.style={font=\small}]
  \node[draw, circle, minimum size=2.2cm] (Cv) at (0,0) {$C_{\lfloor k/2\rfloor}(v)$};
  \node[draw, circle, minimum size=2.2cm] (Cw) at (3,0) {$C_{\lceil k/2\rceil - 1}(w)$};
  \coordinate (v) at (0.8,-0.7);
  \coordinate (w) at (2.1,-0.7);
  \coordinate (u) at (2.1,-1.7);
  \coordinate (x) at (0.78, 0.8);
  \coordinate (y) at (2.22, 0.8);
  \coordinate (z) at (3, -1.1);
  \coordinate (c1) at (0, 0);
  \coordinate (c2) at (3, 0);

  \filldraw[black] (v) circle (2pt) node[below] {$v$};
  \filldraw[black] (w) circle (2pt) node[below left] {$w$};
  \filldraw[black] (u) circle (2pt) node[below] {$u$};
  \draw[thick, magenta, dashed] (u) -- (w) node[midway, left, magenta] {$1$};
  \draw[thick, gray, dotted] (v) -- (c1);
  \draw[thick, gray, dotted] (x) -- (c1);
  \draw[thick, gray, dotted] (y) -- (c2);
  \draw[thick, gray, dotted] (z) -- (c2);
  \draw[thick, magenta, dashed] (w) -- (v) node[midway, below, magenta] {$1$};
  \draw[thick, black] (x) -- (y) node[midway, above, black] {$1$};
  \draw[thick, black] (u) -- (z) node[midway, below, black] {$1$};

\end{tikzpicture}
  \end{minipage}
\caption{
\textbf{Left (Case 2):} \(u\) remains clustered at level \(\lceil k/2\rceil\), and the clusters \(C_{\lfloor k/2\rfloor - 1}(v)\) and \(C_{\lceil k/2\rceil}(u)\) are within distance at most $2$. Rule R3 adds at most two inter‑cluster edges, and cluster‑tree paths from \(v\) and \(u\) complete a path of length \(\le 2k\).
\textbf{Right (Case 3):} \(u\) is unclustered at level \(\lceil k/2\rceil\), so Rule R2\(^*\) adds an edge from \(u\) to \( C_{\lceil k/2\rceil - 1}(w)\). Rule R4 adds an edge between \(C_{\lfloor k/2\rfloor}(v)\) and \(C_{\lceil k/2\rceil - 1}(w)\), and the cluster‑tree paths from \(v\) to \( C_{\lceil k/2\rceil - 1}(w)\) to \( u \) yield a total path of length \(\le 2k\).}

  \label{fig:two-cases}
\end{figure}
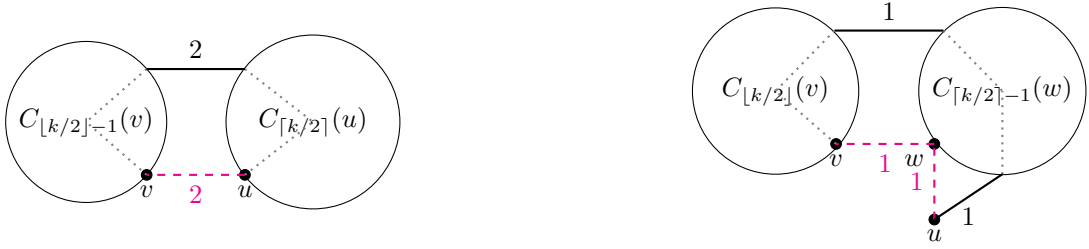

  \item \textbf{Stretch for paths of length 3.}  
    Let \(P=(v,w,z,u)\) be a shortest path of length \(3\).  
    \begin{enumerate}

\item[\textbf{Case 1.}]  
If any of the vertices \(v, w, z, u\) becomes unclustered at some point, assume without loss of generality that it is \(v\) or \(w\). Then, by Lemma~\ref{le:hstretch}, \(\dist_H(v,w) \leq k\). Since \((w,z,u)\) is a shortest path of length 2, the previous part of the proof gives \(\dist_H(w,u) \leq 2k\). Thus,
\[
\dist_H(v,u) \leq \dist_H(v,w) + \dist_H(w,u) \leq k + 2k = 3k.
\]

\item[\textbf{Case 2.}]  
Assume that all vertices \(v, w, z, u\) remain clustered through the relevant levels. Then:

\begin{itemize}
  \item Since \(\dist_G\big(C_{\lfloor k/2 \rfloor - 1}(v),\, C_{\lceil k/2 \rceil}(z)\big) \leq 2\), Rule R3 adds at most two edges between these clusters. We get:
  \[
  \dist_H \big(v,C_{\lceil k/2 \rceil}(z)\big) \leq 2(\lfloor \tfrac{k}{2} \rfloor - 1) + 2 \leq 2 \lfloor \tfrac{k}{2} \rfloor.
  \]
  
  \item Similarly, since \(\dist_G\!\big(C_{\lfloor \tfrac{k}{2} \rfloor - 1}(u),\, C_{\lceil \tfrac{k}{2} \rceil}(z)\big) \le 1\),
Rule R3 adds an inter-cluster edge, yielding
\[
\dist_H\big(u, C_{\lceil \tfrac{k}{2} \rceil}(z)\big)
\;\le\; 2\big(\lfloor \tfrac{k}{2} \rfloor - 1\big) + 1
\;=\; 2\lfloor \tfrac{k}{2} \rfloor - 1.
\]
\end{itemize}

Thus, the path between \(v\) and \(u\) in \(H\) can be approximated by routing through the shared cluster \(C_{\lceil k/2 \rceil}(z)\): we use paths from \(v\) and \(u\) to that cluster, and a short intra-cluster path. Therefore, the total distance between \(v\) and \(u\) is bounded by:
\begin{align*}
\dist_H(v,u) &\leq \dist_H\big(v, C_{\lceil k/2 \rceil}(z)\big) 
+ 2\lceil \tfrac{k}{2} \rceil 
+ \dist_H\big(C_{\lceil k/2 \rceil}(z), u\big) \\
&\leq 2 \lfloor \tfrac{k}{2} \rfloor 
+ 2\lceil \tfrac{k}{2} \rceil 
+ \left( 2 \lfloor \tfrac{k}{2} \rfloor - 1 \right) 
\leq 3k - 1.
\end{align*}

See Figure~\ref{fig:path3}.

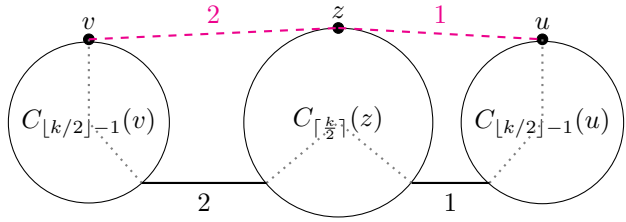
\begin{figure}[h!]
  \centering

\begin{tikzpicture}[every node/.style={font=\small}]

  \node[draw, circle, minimum size=2cm] (Cv) at (0,0) {$C_{\lfloor k/2\rfloor - 1}(v)$};
  \node[draw, circle, minimum size=2.5cm] (Cz) at (3.3,0) {$C_{\lceil \frac{k}{2} \rceil}(z)$};
  \node[draw, circle, minimum size=2cm] (Cu) at (6,0) {$C_{\lfloor k/2\rfloor - 1}(u)$};
  \coordinate (v) at (0,1.1);
  \coordinate (z) at (3.3,1.25);
  \coordinate (u) at (6,1.1);
  \coordinate (vv) at (0.7,-0.8);
  \coordinate (zv) at (2.35,-0.8);
  \coordinate (zu) at (4.27,-0.8);
  \coordinate (uu) at (5.3,-0.8);
  \coordinate (c1) at (0, 0);
  \coordinate (c2) at (3.3, 0);
  \coordinate (c3) at (6, 0);

  \filldraw[black] (v) circle (2pt) node[above] {$v$};
  \filldraw[black] (z) circle (2pt) node[above] {$z$};
  \filldraw[black] (u) circle (2pt) node[above] {$u$};
  \draw[thick, gray, dotted] (v) -- (c1);
  \draw[thick, gray, dotted] (vv) -- (c1);
  \draw[thick, gray, dotted] (zv) -- (c2);
  \draw[thick, gray, dotted] (zu) -- (c2);
    \draw[thick, gray, dotted] (uu) -- (c3);
  \draw[thick, gray, dotted] (u) -- (c3);
  \draw[thick, magenta, dashed] (v) -- (z) node[midway, above, magenta] {$2$};
  \draw[thick, magenta, dashed] (z) -- (u) node[midway, above, magenta] {$1$};
  
  \draw[thick, black] (vv) -- (zv) node[midway, below, black] {$2$};
  \draw[thick, black] (zu) -- (uu) node[midway, below, black] {$1$};

\end{tikzpicture}
\caption{Illustration of Case 2 of the length‑3 path. When all vertices \(v,w,z,u\) remain clustered, we route from \(v\) to the cluster \(C_{\lceil k/2\rceil}(z)\) in at most \(2\lfloor k/2\rfloor\) steps (Rule R3), traverse inside the cluster with at most $2\lceil \tfrac{k}{2} \rceil$ steps, and then go from \(C_{\lceil k/2\rceil}(z)\) to \(u\) in at most \(2\lfloor k/2\rfloor-1\) steps. Altogether this yields \(\dist_H(v,u)\le3k-1\).}

  \label{fig:path3}
\end{figure}

\end{enumerate}

  \item \textbf{General paths \(\dist(v,u)=d\ge2\).}  
    Decompose a shortest path of length \(d\) into segments of length \(2\) or \(3\). Using the bounds \(\le 2k\) for length-2 and \(\le 3k\) for length-3, concatenation gives \(\dist_H(v,u)\le k\cdot d\).
\end{enumerate}
This completes the proof that \(H\) is a \(k\)-hybrid spanner.
\end{proof}


\bibliographystyle{unsrt}
\bibliography{mybib}

\pagebreak

\section{Constructing a $ 2k + O(d\ld) $ Spanner in $O(m \sqrt{n})$ Time} \label{sec:apndx}

In this section, we show how to modify the algorithm from Section~\ref{sec:ab-spa} to achieve a running time of \(O(m \sqrt{n})\). Before invoking Algorithm \spaOne, we apply a preprocessing step taken from~\cite{roundtrip_spanners19}, whose goal is to reduce the maximum degree of the graph to \(O(m/n)\).

\paragraph{The preprocessing step.}
Let \(\delta = \lceil m/n \rceil\). For each vertex \(v \in V\), we replace all of its incident edges with a balanced \(\delta\)-ary tree \(T_v\) rooted at \(v\), where all internal edges of the tree have weight \(0\). Each leaf \(u \in T_v\) is responsible for \(\delta'(u) < \delta\) of the original edges incident to \(v\); that is, \(u\) has \(\delta'(u)\) edges to \(\delta'(u)\) distinct neighbors of \(v\), each corresponding to one of the original edges. These edges retain weight \(1\), while the internal tree edges remain with weight \(0\). Every original edge is represented by some leaf, so
\[
\sum_{\text{leaf } u \in T_v} \delta'(u) = \deg(v).
\]
Let the resulting graph be \(G' = (V', E', \omega)\), where \(V \subseteq V'\), and \(\omega: E' \to \{0,1\}\) specifies the edge weights.
The following lemma is proven in~\cite{roundtrip_spanners19} for directed graphs; the same result holds for undirected graphs as well.

\begin{theorem}[Lemma 3.1 in~\cite{roundtrip_spanners19}]
After applying the transformation to obtain \(G' = (V', E', \omega)\), the following properties hold:
\begin{itemize}
    \item \(|V'| = O(n)\) and \(|E'| = O(m + n)\).
    \item All pairwise distances (between pairs of vertices in $G$) in $G'$ and in $G$ are the same.
    \item For every \(v' \in V'\), \(\deg(v') =  O(\tfrac{m}{n}) \).
\end{itemize}
Moreover, this transformation can be computed in \(O(m + n)\) time.
\end{theorem}

We omit the proof and refer the reader to~\cite{roundtrip_spanners19} for details. Once this transformation is applied, we run Algorithms \spaOne and \spaTwo on \(G'\) to obtain a spanner \(H' \cup H'_s \subseteq G'\) with the desired stretch guarantee. We then contract the added vertices in \(H'\) to return a spanner \(H \subseteq G\)\footnote{All distances and balls in \(G'\) are defined using
edge weights, and BFS is replaced by 0--1 BFS. We include all zero-weight edges in \(H'\) from the outset. They disappear upon contraction. The second-phase stretch analysis applies to unit-weight edges, while zero-weight edges are preserved exactly.}.

\medskip
The full algorithm works as follows. Given a graph \(G = (V, E)\) and two integers \(d \leq k\), where \(d\) is a power of 2 and divides \(k\), the algorithm performs:

\begin{itemize}
    \item[$\bullet$] It applies the preprocessing step to transform \(G\) into \(G' = (V', E', \omega)\), as described above.
    
    \item[$\bullet$] Let \((H', V'_\ld) = \textsc{SpannerPartOne}(G', k, d)\).
    
    \item[$\bullet$] Let \(H'_s = \textsc{SpannerPartTwo}(G'[V'_\ld], k/d, k)\).
    
    \item[$\bullet$] Let \(H \subseteq G\) be the result of contracting back the auxiliary vertices of \(H' \cup H'_s\) to the original vertex set \(V\).
    
    \item[$\bullet$] Lastly, it returns \(H\).
\end{itemize}

We now prove Theorem~\ref{th:ab-spa}, beginning with the stretch guarantee of \(H\). This follows immediately, since distances between vertices in \(V\) are preserved throughout the transformation: distances in \(G\) and \(G'\), as well as in \(H\) and \(H' \cup H'_s\), are identical when restricted to the original set \(V\).

Next, we bound the size of \(H\) and show that \(|E(H)| = O(n^{1+1/k})\). The first phase of the algorithm is identical to the one analyzed in the proof of Lemma~\ref{le:size1}. There, we argued that in iteration \(i\), the algorithm adds at most \(r_i\) edges to \(H'\) between any pair of vertices \(c, c' \in C_i\) that are at distance at most \(r_i\).

While this argument may appear problematic here—since shortest paths in \(G'\) may include many edges of weight \(0\)—we note that all such zero-weight edges are removed in the final contraction step that maps \(H' \cup H'_s\) back to a subgraph of \(G\). Therefore, only edges of weight \(1\) contribute to the final spanner \(H\), and the original analysis applies.

The second phase is exactly as in Claim \ref{cl:size2}, and hence its edge bound remains the same.

To complete the proof of Theorem~\ref{th:ab-spa}, it remains to bound the running time. We establish this with the following claim:

\begin{claim}
    The total running time of the algorithm is \(O(m \sqrt{n})\) in expectation.
\end{claim}

\begin{proof}
The preprocessing step transforms the input graph \(G\) into a graph \(G'\) with \(O(n)\) vertices and maximum degree \(O(m/n)\), in \(O(m + n)\) time. Due to the bounded degree, a BFS from any vertex \(x\) up to distance \(r\) in every subgraph \(G'' \subseteq G'\) takes \(O\left( \frac{m}{n} \cdot |B_{G''}(x, r)| \right)\) time.

We first analyze the running time of the first phase of the algorithm. In iteration \(i\), the algorithm selects (in expectation) \(|C_i| = O(n^{1 - 1/2^{i+1}})\) centers and performs a BFS from each center to depth \(r_i\). For each center \(c \in C_i\), the expected time for the BFS is \(O\left( \frac{m}{n} \cdot |B_i(c, r_i)| \right)\). By Claim~\ref{cl:size_gen}, we have \(\mathbb{E}[|B_i(c, r_i)|] = O(n^{1/2^i})\), so the total expected time for all BFS traversals in iteration \(i\) is:

\begin{equation*}
\mathbb{E}[|C_i|] \cdot \frac{m}{n} \cdot \mathbb{E}[|B_i(c, r_i)|]
= O\left(n^{1 - 1/2^{i+1}} \cdot \frac{m}{n} \cdot n^{1/2^i}\right)
= O\left(m \cdot n^{1/2^i - 1/2^{i+1}}\right)
= O\left(m \cdot n^{1/2^{i+1}}\right).
\end{equation*}

Additionally, each iteration includes a BFS from the set \(C_i\), which takes \(O(m)\) time, and an update to the vertex set \(V_{i+1}\), which takes \(O(n)\) time.

Summing over iterations \(i = 0\) to \(\log d - 1\), the total expected running time of the first phase is:

\begin{equation*}
\sum_{i=0}^{\log d - 1} \left( O\left(m \cdot n^{1/2^{i+1}}\right) + O(m) + O(n) \right).
\end{equation*}

Since \(n^{1/2^{i+1}}\) decreases rapidly, the dominant term is \(n^{1/2}\) at \(i = 0\), and the total sum is dominated by \(O(m \sqrt{n})\). Thus, the first phase takes \(O(m \sqrt{n})\) time in expectation.

In the second phase, the dominant part of the running time is computing the budget \(b(x)\) for each vertex \(x \in V_s\), which requires running a BFS from each such vertex up to distance \(R\). The total expected running time is:

\begin{equation*}
\mathbb{E}\left[ \sum_{x \in V_s} O\left( \frac{m}{n} \cdot |B_s(x, R)| \right) \right]
= O\left( \frac{m}{n} \sum_{x \in V_s} \mathbb{E}[|B_s(x, R)|] \right)
\leq O\left( m \cdot n^{R/k} \right),
\end{equation*}

where the last inequality follows from the \((R, k)\)-sparsity property. Substituting \(d \geq 2\) and \(R = \tfrac{k}{d}\), we get \(n^{R/k} = n^{1/d} \leq \sqrt{n}\), and thus the second phase runs in \(O(m \sqrt{n})\) time.

Combining both phases, the total expected running time of the algorithm is \(O(m \sqrt{n})\).

\end{proof}



\end{document}